\pdfoutput=1
\documentclass[a4paper,10pt,twocolumn,5p]{elsarticle}
\usepackage{booktabs}
\usepackage{threeparttable}
\usepackage[numbers]{natbib}
\usepackage{amsmath}
\usepackage{wasysym}
\usepackage{booktabs}
\usepackage[caption=false,font=footnotesize]{subfig}
\usepackage{algpseudocode,algorithm}
\makeatletter
\algrenewcommand\ALG@beginalgorithmic{\footnotesize}
\makeatother
\usepackage{dblfloatfix}
\usepackage[dvipsnames]{xcolor}
\usepackage{graphicx}
\usepackage{amssymb}
\usepackage{amsthm}
\usepackage{amsfonts}
\usepackage{enumitem}
\usepackage{array}
\usepackage{hyperref}
\hypersetup{colorlinks	=true , %
			pdftitle 	= {Wideband Modal Orthogonality: A New Approach for Broadband DOA Estimation} , %
			pdfauthor 	= {S.Amirsoleimani}}

\newcommand{\approptoinn}[2]{\mathrel{\vcenter{
			\offinterlineskip\halign{\hfil$##$\cr
				#1\propto\cr\noalign{\kern2pt}#1\sim\cr\noalign{\kern-2pt}}}}}

\theoremstyle{plain} 
\theoremstyle{plain} \newtheorem{theorem}{Theorem}
\theoremstyle{remark}
\theoremstyle{definition}\newtheorem{definition}{Definition}
\theoremstyle{definition}\newtheorem{corollary}{Corollary}

\definecolor{MyBlue}{rgb}{0.0,0.0,0.0} 

\journal{Journal of Signal Processing}
\begin{document}	
\begin{frontmatter}
	\title{Wideband Modal Orthogonality: A New Approach for Broadband DOA Estimation}
	%
	\author{Shervin~Amirsoleimani}
	\ead{amirsoleimani@ut.ac.ir}
	\author{Ali~Olfat}
	\ead{aolfat@ut.ac.ir}
	\address{Signal Processing and Communication Systems Laboratory, University of Tehran, Tehran, Iran.}
	
	\begin{abstract}
		Wideband direction of arrival (DOA) estimation techniques for sensors array have been studied extensively in the literature. Nevertheless, needing prior information on the number and directions of sources or demanding heavy computational load makes most of these techniques less useful in practice.
		In this paper, a low complexity subspace-based framework for DOA estimation of broadband signals, named as wideband modal orthogonality (WIMO), is proposed and accordingly two DOA estimators are developed.
		First, a closed-form approximation of spatial-temporal covariance matrix (STCM) in the uniform spectrum case is presented. The eigenvectors of STCM associated with non-zero eigenvalues are 
		modal components of the wideband source in a given bandwidth and direction. 
		WIMO idea is to extract these eigenvectors at desired DOAs from the approximated STCM and test their orthogonality to estimated noise subspace.
		In the non-uniform spectrum case, WIMO idea can be applied by approximating STCM through numerical integration.
		Fortunately, STCM approximation and modal extraction can be performed offline.
		WIMO provides DOA estimation without the conventional prerequisites, such as spectral decomposition, focusing procedure and, a priori information on the number of sources and their DOAs. 
		Several numerical examples are conducted to compare the WIMO performance with the state-of-the-art methods. 
		Simulations demonstrate that the two proposed DOA estimators achieve superior performance in terms of probability of resolution and estimation error along with orders of magnitude runtime speedup.
		
	\end{abstract}
	
\end{frontmatter}
\section{Introduction}\label{sec.introduction}
Wideband is attributed to those signals whose power spectrum occupies a fractional bandwidth of larger than 20$\%$ \cite{Nikookar2009}. This type of signals arise in many applications such as radar, passive sonar, microphone arrays, seismology and high rate wireless communication systems \cite{Chellappa2014}. 
In the wideband scenario, unlike the narrowband, each frequency bin carries different information regarding source direction of arrival (DOA) angles \cite{Tuncer2009}.
Accordingly, observations are commonly decomposed into multiple subbands, such that in each subband, the narrowband assumption holds. 
The processing methods are generally categorized as incoherent \cite{Kailatha1984} and coherent \cite{Wang1985}. In both methods the covariance matrix at all subbands are estimated. While incoherent method applies narrowband DOA estimation at each subband, the coherent solution 
combines all subband covariances coherently and applies DOA estimation on the final covariance matrix.
The incoherent approach has poor performance in low SNR, while 
most of the coherent schemes 
require a priori information on the number of sources and their DOAs to initiate focusing procedure. 
\color{MyBlue}
There is also some extensions on CSSM which could partially reduce some of its shortcomings using beamforming invariant techniques. See \cite{Lee1994,ward1998broadband,Sathish2006} for details.
\color{black}

Another algorithm which effectively overcomes the coherent approach shortcomings is TOPS \cite{Yoon2006a}. TOPS tests the orthogonality of the projected signal subspace and the noise at each DOA and frequency bin. It does not require preliminaries of the coherent scheme but suffers from spurious peaks at all SNR levels.
Performance of TOPS has been improved further in the subsequent developments \cite{Okane2010,Shaw2016}. 

Maximum likelihood (ML) approaches have also been studied for wideband source localization \cite{Chen2002,Mada2009,Lu2011b}. ML-based solutions are statistically optimum, but lead to nonlinear optimization problem in the presence of multiple sources.

Other methods utilizing time delayed samples have been developed in parallel \cite{buckley1988broad,Yin2018b,Agrawal2000}. 
BASS-ALE \cite{buckley1988broad} incorporates temporally delayed samples in the observation vectors. It first estimates spatial-temporal covariance matrix (STCM), then tests the orthogonality of the array manifold with the estimated noise subspace at each frequency bin. BASS-ALE requires only single eigenvalue decomposition (EVD) for the whole band but suffers from the intrinsic loss of incoherency.
Agrawal and Prasad \cite{Agrawal2000} has proposed two spatial-only approaches which are based on subspace of columns of wideband source's covariance matrix. In the first approach  an exhaustive search among all $\left\lbrace \theta_1,\cdots,\theta_\mathcal{K}\right\rbrace$ and also sources' number $\mathcal{K}$ is required, which is computationally intractable in real applications. The second approach introduced in \cite{Agrawal2000} is a feasible method but has lower performance compared to the first one. There is also no specific solution to distinguish eigenvectors corresponding to noise subspace.
\color{black}
STEP \cite{Yin2018b} method also falls in this class. It calculates steered covariance matrix at each test angle. Then tests the orthogonality of signal subspace with the noise. STEP requires a heavy computational load and is highly sensitive to signal subspace order selection.

With the advent of sparse representation (SR) framework, it has also been applied to wideband DOA estimation problem \cite{Yang2017,malioutov2005sparse,liu2011direction,Hu2012a,Amirsoleimani2019,He2015}. $\ell_1$-SVD \cite{malioutov2005sparse} represents each subband data in an overcomplete steering vector dictionary, then forms a joint sparse representation problem according to identical spectral support of wideband sources.
A similar approach has been followed in W-SpSF \cite{He2015}, but with utilizing the modified subband covariance matrix as the observation vector.
On the other hand, W-CMSR \cite{liu2011direction} and W-LASSO \cite{Hu2012a} have pursued a time-domain approach. They directly represent the array covariance matrix in temporally delayed versions of source correlation function; thus, no subband processing is involved. 
The main drawback of SR-based schemes are the selection of the regularization parameter, which severely affects the overall performance. They also suffer from heavy computational burden due to the sparse solution recovery procedure.

In this paper, a fast solution for wideband direction of arrival estimation is introduced. 
The proposed approach falls into the category of subspace-based methods. It extracts the noise subspace from the estimated spatial-temporal covariance matrix. On the other hand, a closed-form approximation of the STCM 
is derived for uniform power spectral density (PSD).
Through this approximation the eigenvectors of STCM corresponding to signal subspace can be computed for desired DOA grid points. These eigenvectors are named as modal components of a wideband source at direction $\theta$ and $f\in [f_l, f_h]$. 
Depending on the number of components participating in the orthogonality test, two wideband DOA estimators are proposed. Pure wideband modal orthogonality (p-WIMO) applies the whole approximated matrix as the signal subspace, while 1-WIMO only uses the most powerful signal mode. This mode is named as \textit{generalized steering vector} since it converges to the array steering vector as the bandwidth tends to zero.

WIMO requires a single eigenvalue decomposition (EVD) for the whole band, similar to the BASS-ALE approach \cite{buckley1988broad}, however, in contrast to BASS-ALE, it does not involve incoherent amplitude summation on spatial spectra. In addition, the main computations in WIMO appertain to calculation of approximation matrix and modal components retrieval, which is performed offline. 
Consequently, modal components for all test angles can be saved in a database prior to any computation. 
Furthermore, no spectral decomposition or subband processing is introduced; hence, WIMO process time is independent of signal bandwidth.
Finally, no a priori information on the number of sources and their DOAs are necessary and no tedious focusing procedure is involved in WIMO.
All of these advantages, make WIMO an efficient broadband DOA estimator in terms of computational cost and practical implementation.

WIMO can also be applied for non-uniform PSD. It is shown that in the general case, STCM approximation can be done through a numerical integration which only adds some more offline computational load. 
It is also shown that applying uniform PSD assumption to signal with intrinsically non-uniform PSD would lead to a mismatch loss that can be reduced with appropriate uniform bandwidth selection.
Of course, closed-form approximation would be obtained for some special form of power spectral density which is not the subject of this article.

The rest of the paper is organized as follows. Spatial-temporal observation model of array data is reviewed in Section~\ref{sec.signal_model}.
The proposed wideband DOA estimator using the modal orthogonality concept is introduced in Section~\ref{sec.wideband}. This section also contains complementary issues on STCM approximation matrix structure, asymptotic behavior of STCM and broadband signal subspace dimension estimation, algorithms' complexity analysis and WIMO extension in non-uniform PSD case.
The simulation results and comparison with state-of-the-art methods are presented in section~\ref{sec.simulation}. Finally, conclusions are drawn in Section~\ref{sec.conclusion}.
\section{Spatial-Temporal observation model}\label{sec.signal_model}
Assume a plane incident wave, $x(t)$, impinging on an array of $N_S$ omnidirectional sensors. Considering the first antenna as the time reference, the $k$th sensor observation can be formulated as,
\begin{equation} \label{equ.meas_model1}
y_k(t) = x(t-\tau_k) \quad k=0,1,\cdots,N_S-1 \quad \tau_0=0
\end{equation}
where $\tau_k$ corresponds to $k$th sensor delay. For a linear array $\tau_k = -z_k \sin\theta/c$ and in the general 3D case,
\begin{equation} 
\tau_k = -\frac{\mathbf{u}^T(\theta,\phi)\mathbf{p}_k}{c} \quad , \quad
\mathbf{u}(\theta,\phi) \triangleq \left[\begin{array}{c}
\cos\theta \cos\phi \\
\cos\theta \sin\phi \\
\sin\theta
\end{array}\right]
\end{equation}
where $\mathbf{p}_k = [x_k \: , y_k \: , z_k]^T$ stands for the $k$th sensor position in Cartesian coordinate and $c$ is the wave propagation velocity. 
By denoting $X(f)$ as the Fourier transform of $x(t)$, \eqref{equ.meas_model1} can be rewritten as,
\begin{equation}\label{equ.IFT}
x(t-\tau_k) = \int X(f)e^{j2\pi f (t-\tau_k)}\cdot df
\end{equation}
In the narrowband scenario, 
$X(f) \simeq \delta(f-f_c)$ and the integral in \eqref{equ.IFT} is simplified as a phase-shift, $x(t-\tau_k) \approxeq e^{-j2\pi f_c \tau_k}x(t)$. Therefore, the observation vector $\mathbf{y}(t)$ with entries defined in \eqref{equ.meas_model1} can be expressed as,
\begin{equation} \label{equ.narrow_model}
\mathbf{y}(t) = \mathbf{a}(\theta,f_c) x(t)
\end{equation}
where $\mathbf{a}(\theta,f_c) = [e^{-j2\pi f_c \tau_0(\theta)} , \cdots , e^{-j2\pi f_c \tau_{N_S-1}(\theta)}]^T$ is the array steering vector. The above discussion considers a rank-1 model where each independent source spans an one-dimensional signal subspace, however, in wideband scenario a single source captures multiple frequency bins and naturally, lies in a subspace with larger dimensions. 

Spatial-temporal observation vector is constructed by stacking the delayed temporal samples of the array snapshots. 
This leads to better discrimination of wideband sources due to their different instantaneous frequency and also larger subspace dimensions \cite{buckley1988broad}. 
Using $m$ temporally delayed samples of each sensor, spatial-temporal observation vector $\tilde{\mathbf{y}}(t)$ is written as,
\begin{equation} \label{equ.stom}
\begin{aligned}
\tilde{\mathbf{y}}(t) = &\left[y_0(t-(m-1)dt),\cdots,y_0(t-dt),y_0(t),\cdots\right. \\
&\left. y_{N_S-1}(t-(m-1)dt),\cdots,y_{N_S-1}(t)\right]^T _{mN_s\times 1}
\end{aligned}
\end{equation}
Defining $\mathbf{t} \triangleq \left[0,\; dt,\;\cdots,\; (m-1)dt\right]^T$ and using \eqref{equ.IFT}, $\tilde{\mathbf{y}}(t)$ can be expressed by inverse Fourier transform as,
\begin{equation} \label{equ.dis_obs_model}
\begin{aligned}
\tilde{\mathbf{y}}(t)
&=\int X(f)\left[\begin{array}{c}
e^{j2\pi f\mathbf{t}} \\
e^{j2\pi f(-\tau_1 \mathbf{1} + \mathbf{t})}\\
\vdots \\
e^{j2\pi f(-\tau_{N_S - 1} \mathbf{1} + \mathbf{t})}
\end{array}\right]\cdot df \\
&=\int X(f) \mathbf{g}(f,\theta) \cdot df
\end{aligned}
\end{equation}
where $\mathbf{1}_{m\times 1} = [1,\cdots , 1]^T$ and $dt$ is the sampling interval. The $\mathbf{g}(f,\theta)$ vector represents the spatial-temporal observation model for a single tone at frequency $f$ impinging on the array from angle $\theta$.

\subsection{spatial-temporal covariance matrix (STCM)}
Using \eqref{equ.stom}, Spatial-temporal covariance matrix (STCM) is defined as,
\begin{equation} \label{equ.STCM_def}
\mathbf{S}_{\tilde{\mathbf{y}} \tilde{\mathbf{y}}} \triangleq 
\mathbb{E}\left\lbrace\tilde{\mathbf{y}}(t) \tilde{\mathbf{y}}(t)^H\right\rbrace
\end{equation}
Assuming $\mathcal{K}$ uncorrelated sources with center frequency $f_c$ and bandwidth $B$ and spatial-temporal white additive noise, $\mathbf{S}_{\tilde{\mathbf{y}} \tilde{\mathbf{y}}}$ can be written as,
\begin{equation} \label{equ.STCM2}
\mathbf{S}_{\tilde{\mathbf{y}} \tilde{\mathbf{y}}} = 
\sum_{k=1}^{\mathcal{K}} \mathbf{S}(\theta_k,f_c,B) + \sigma_n^2 \mathbf{I}_L
\end{equation}
where $\mathbf{S}(\theta,f_c,B)$ is the STCM of a single wideband source located at $\theta$, $\sigma_n^2$ is the variance of noise and $L=mN_S$ is the length of observation vector $\tilde{\mathbf{y}}(t)$. 
Due to limited time-bandwidth product of sources, $\mathbf{S}(\theta,f_c,B)$ has a low rank structure and for a modest number of temporally delayed samples $m$, the dimension of broadband signal subspace is less than the observations space dimension $L$ \cite{M.Buckley1986,buckley1988broad}. Accordingly, by proper selection of $m$ a noise-only subspace exists.
Let $P$ denotes the number of eigenvectors belonging to the signal subspace, 
using eigenvalue decomposition (EVD), \eqref{equ.STCM2} is rewritten as
\begin{equation}\label{equ.Syy_EVD}
\mathbf{S_{\tilde{\mathbf{y}} \tilde{\mathbf{y}}}} = 
\mathbf{U}_s \mathbf{\Lambda}_s \mathbf{U}_s^H + \mathbf{U}_n \mathbf{\Lambda}_n \mathbf{U}_n^H
\end{equation}
Denote $\lambda _1 \ge \lambda_2\ge\dots \ge\lambda_L$ as the eigenvalues of $\mathbf{S}_{\tilde{\mathbf{y}} \tilde{\mathbf{y}}}$, $\mathbf{\Lambda}_s = \text{diag}(\lambda_1 , \cdots , \lambda_P)$ and 
$\mathbf{U}_s = [\mathbf{u}_1 , \cdots , \mathbf{u}_P]$ as eigenvalues and eigenvectors corresponding to signal subspace respectively, and $\mathbf{\Lambda}_n = \text{diag}(\lambda_{P+1} , \cdots , \lambda_{L})$ and $\mathbf{U}_n = [\mathbf{u}_{P+1} , \cdots , \mathbf{u}_{L}]$ as the eigenvalues and eigenvectors corresponding to noise subspace. 
Then, the estimation of STCM is obtained using $M$ time snapshots as,
\color{black}
\begin{equation} \label{equ.Syy_hat}
\hat{\mathbf{S}}_{\tilde{\mathbf{y}} \tilde{\mathbf{y}}} = 
\frac{1}{M-m+1}\sum_{t=m}^{M} \tilde{\mathbf{y}}(t) \tilde{\mathbf{y}}(t)^H
\end{equation}
Similarly, $\hat{\mathbf{U}}_s$ and $\hat{\mathbf{U}}_n$ refer to the estimated eigenvectors for signal and noise subspace, by applying EVD to $\hat{\mathbf{S}}_{\tilde{\mathbf{y}} \tilde{\mathbf{y}}}$ respectively.
\subsection{Spatial-spectrum transform}
The spatial-temporal model \eqref{equ.dis_obs_model} has been used in previous works \cite{Kailatha1984,buckley1988broad}, 
but $\mathbf{g}(f,\theta)$ itself can be used as a low-resolution correlation estimator for spatial-spectrum analysis of the array data. This interpretation is formulated as a \textit{Spatial-Spectrum transform}.
\begin{definition}
	\textit{Spatial-Spectrum transform of order $m$}, for the spatial-temporal observation vector $\tilde{\mathbf{y}}\in\mathbb{C}^L$, denoted by $Y_{\mathcal{SS}^m}(f,\theta)$, is a linear transform $t \mapsto [f,\theta]$ defined as:
	\begin{equation} \label{equ.SOT} 
	Y_{\mathcal{SS}^m}(f,\theta) \triangleq \sum_{k} \tilde{y}_k e^{-j2\pi h_k f}
	\end{equation}
	where $\tilde{\mathbf{y}}$ is defined in \eqref{equ.stom} and $m$ is the temporal lag order.
	$h_k$ is an implicit function of elevation angle $\theta$, temporal sample vector $\mathbf{t}$ and the array elements' location $\mathbf{p}_k$. According to \eqref{equ.dis_obs_model}, $h_k$ can be expressed as,
	\begin{equation} \label{equ.hk}
	h_k = -\tau_{\lfloor \frac{k}{m}\rfloor} + t_{\mathtt{mod}(k,m)} \quad k\in \lbrace 0,\cdots,mN_S - 1\rbrace
	\end{equation}
	where $\lfloor \cdot \rfloor$ is the floor operator and $\mathtt{mod}(\cdot,\cdot)$ denotes the modulo operator.
	\color{black}
\end{definition}
For $m$=1, SS-Transform reduces to the conventional beamformer, and with $N_S$=1 it is simplified to discrete Fourier transform. This interpretation will be used in the next section.
\section{Proposed Wideband DOA Estimator}\label{sec.wideband}
In this section, we consider the localization of wideband sources using the spatial-temporal covariance matrix.
Let $\mathbf{u}_i(\theta,f_c,B)$ denote the eigenvector of $\mathbf{S}(\theta,f_c,B)$ corresponding to $i$-th non-zero eigenvalue arranged in decreasing order. 
Then, in the presence of $\mathcal{K}$ sources, each eigenvector corresponding to the noise subspace of STCM (columns of $\mathbf{U}_n$ in \eqref{equ.Syy_EVD}) 
is orthogonal to $\mathbf{u}_{i}(\theta_k,f_c,B)$, i.e., we can write
\begin{equation} \label{equ.wimo-theory}
\mathbf{u}_{i}(\theta,f_c,B) \in \mathrm{null}(\mathbf{U}_n)
\end{equation}
for $\theta \in \{\theta_1 , \cdots , \theta_\mathcal{K}\}$ and $i \in \{1,\cdots,\mathrm{rank}(\mathbf{S}(\theta,f_c,B))\}$.
Since in the wideband case $\mathrm{rank}(\mathbf{S}(\theta,f_c,B))$ is larger than one, for each direction there are several $\mathbf{u}_{i}(\theta,f_c,B)$ vectors which \eqref{equ.wimo-theory} can be applied to estimate the sources' directions. 
The idea is to approximate $\mathbf{u}_{i}(\theta,f_c,B)$ for desired $\theta$ points and then test \eqref{equ.wimo-theory} with $\hat{\mathbf{U}}_n$ estimated from the observations. 
As in \cite{buckley1988broad}, we use the term \textit{modal component} for $\mathbf{u}_{i}(\theta,f_c,B)$ vectors and will denote it occasionally by $\mathbf{u}_i$ for simplicity. Also, $\breve{ }$ superscript is used as approximate of a quantity; e.g. $\breve{\mathbf{S}}(\theta,f_c,B)$ denotes the approximate of $\mathbf{S}(\theta,f_c,B)$ and $\breve{\mathbf{u}}_{i}(\theta,f_c,B)$ stands for an approximate of $\mathbf{u}_{i}(\theta,f_c,B)$.


\subsection{STCM Approximation} \label{sec.stcm-app}
\color{MyBlue}
Spatial covariance matrix approximation was first presented in \cite{Agrawal2000}. In this subsection we extend the approximation for spatial-temporal covariance matrix. 
\color{black}
By substitution of \eqref{equ.dis_obs_model} in STCM definition \eqref{equ.STCM_def}, for a single source located at $\theta$ with spectral content $f\in [f_l, f_h]$, we obtain,
\begin{equation} \label{equ.obs_covariance}
\begin{aligned}
\mathbf{S}_{\tilde{\mathbf{y}} \tilde{\mathbf{y}}} = 
\iint_{f_l}^{f_h}\mathbb{E}\lbrace X(u)X^*(v)\rbrace
\mathbf{g}(u,\theta)\mathbf{g}^H(v,\theta)\cdot du dv
\end{aligned}
\end{equation}
where $\mathbb{E}\lbrace X(u)X^*(v)\rbrace$ is the cross-correlation of Fourier transform of source signal $x(t)$. 
If the observation time $mdt=T$ at each spatial-temporal vector, is larger than the signal coherence time $\tau_0 = 1/B$, the frequency bins of the Fourier transform become uncorrelated \cite[p.~315]{Chellappa2014}. It means that if $mB/f_s\gg 1$, we can write,

\begin{equation}\label{equ.FC_approx}
\mathbb{E}\lbrace X(u)X^*(v)\rbrace \approxeq S(u)\delta(u-v)
\end{equation}
where $\delta(t)$ is the Dirac Delta function and $S(f)$ is the PSD of the source signal.
Substituting \eqref{equ.FC_approx} in \eqref{equ.obs_covariance},
\begin{align} \label{equ.int_approx}
\mathbf{S}_{\tilde{\mathbf{y}} \tilde{\mathbf{y}}}
\approxeq \int_{f_l}^{f_h} S(u)\cdot\mathbf{g}(u,\theta)\mathbf{g}^H(u,\theta) \cdot du 
\end{align}
Recalling the definition of the correlation function and its relation with PSD,
\begin{equation}
r_x(\tau) \triangleq \mathbb{E}\lbrace x(t)x^*(t-\tau)\rbrace = 
\int_{f_l}^{f_h} S(f) e^{j2\pi f\tau} \cdot df
\end{equation}
Then according to \eqref{equ.dis_obs_model}, $g_k = \exp\left(j2\pi h_k f\right)$, STCM approximation entries in \eqref{equ.int_approx} are in fact rearranged samples of correlation function. Denoting $(k,l)$-th entry of $\mathbf{S}_{\tilde{\mathbf{y}} \tilde{\mathbf{y}}}$ with $s_{k,l}$,
\begin{equation} \label{equ.corr-function}
s_{k,l} = r_x(h_k - h_l)
\end{equation}
Assuming a uniform PSD (see Section~\ref{sec.non-uniform} for the general non-uniform case), then
\begin{equation} \label{equ.temp1}
S(f) = 
\begin{cases}
\frac{\sigma_x^2}{f_h - f_l} & f\in [f_l, f_h] \\
0 & \text{o.w.}
\end{cases}
\end{equation}
Substituting \eqref{equ.temp1} in \eqref{equ.int_approx}, the following approximation for $s_{k,l}$ in the uniform PSD case is derived,
\begin{equation} \label{equ.s_kl_p}
s_{k,l} \approxeq \frac{\sigma_x^2}{B} \frac{\sin(\pi B (h_k - h_l))}{\pi(h_k - h_l)} e^{j2\pi f_c(h_k - h_l)}
\end{equation}
where $f_c = \frac{f_l + f_h}{2}$ and $B=f_h - f_l$.
Let $\breve{s}_{k,l}$ denotes the normalized entry of the approximated STCM:
\begin{equation} \label{equ.smat}
\breve{\mathbf{S}}(\theta,f_c,B) = [\breve{s}_{k,l}] \quad , \quad 
s_{k,l} \approxeq \sigma_x^2 \breve{s}_{k,l}
\end{equation}
which represents the approximation of STCM for a single source at $\theta$ with spectral content $f\in [f_c - \frac{B}{2} , f_c + \frac{B}{2}]$. 
Consequently, for $\mathcal{K}$ sources the following approximation for $\mathbf{S_{\tilde{\mathbf{y}} \tilde{\mathbf{y}}}}$ holds,
\begin{equation}
\mathbf{S_{\tilde{\mathbf{y}} \tilde{\mathbf{y}}}} \approxeq \sum_{k=1}^{\mathcal{K}} \sigma_{x,k}^2 
\breve{\mathbf{S}}(\theta_k , f_{c,k} , B_k)
\end{equation}
Also, $\breve{\mathbf{S}}$ and $\breve{\mathbf{S}}(\theta)$ will be occasionally used instead of $\breve{\mathbf{S}}(\theta,f_c,B)$ for simplicity.

\subsection{Generalized steering vector} \label{sec.gsv}
By providing an approximation for STCM, we introduce the generalized steering vector as the most powerful modal component of the broadband source.
\begin{definition} \label{def.gsv}
	\textit{Generalized Steering Vector} (GSV), $\breve{\mathbf{u}}_1(\theta,f_c,B)$, for a broadband source at direction $\theta$ and center frequency $f_c$ with bandwidth $B$, is the eigenvector of the $\breve{\mathbf{S}} = [\breve{s}_{k,l}]$ matrix associated with the largest eigenvalue.
\end{definition}
	
\begin{corollary}
Since the matrix $\breve{\mathbf{S}}$ is Hermitian, $\breve{s}_{k,l} = \breve{s}_{l,k}^*$ , based on the \textit{Rayleigh quotient theorem} \cite{Horn2013}, the generalized steering vector $\breve{\mathbf{u}}_1$ can be expressed as,
\begin{equation}\label{equ.Rayleigh}
\breve{\mathbf{u}}_1(\theta,f_c,B) \triangleq \underset{\mathbf{x}}{\text{argmax}} 
\; \frac{\mathbf{x}^H \breve{\mathbf{S}}(\theta,f_c,B)\mathbf{x}}{\|\mathbf{x}\|_2^2}
\end{equation}
\end{corollary}

\begin{corollary}
	The matrix $\breve{\mathbf{S}}$ can be expressed as the Hadamard product of two matrices as,
	\begin{align}
	&\breve{\mathbf{S}}(\theta,f_c,B) = \breve{\mathbf{S}}^N(\theta,f_c) \circ \breve{\mathbf{S}}^W(\theta,B) \\
	&\breve{s}^N_{k,l} = g_k g_l^* \quad , \; \breve{s}^W_{k,l} = \frac{\sin(\pi B (h_k - h_l))}{\pi B (h_k - h_l)} \label{equ.sn_sw}
	\end{align}
	where the Hadamard product of $\mathbf{A}$ and $\mathbf{B}$ is the element-wise product matrix $\mathbf{A}\circ\mathbf{B} = [a_{k,l}b_{k,l}]$ and N and W superscripts denote narrowband and wideband respectively. 
\end{corollary}
\color{MyBlue}
The following theorem clarifies the relation between the aforementioned narrowband and wideband part of the $\breve{\mathbf{S}}$ matrix.
\color{black}
\begin{theorem}\label{the.1}
Let $\lbrace \sigma_i(\breve{\mathbf{S}})\rbrace_{i=1}^L$, $\lbrace \sigma_i(\breve{\mathbf{S}}^N)\rbrace_{i=1}^L$ and $\lbrace \sigma_i(\breve{\mathbf{S}}^W)\rbrace_{i=1}^L$ be the eigenvalues of $\breve{\mathbf{S}}$, $\breve{\mathbf{S}}^N$ and $\breve{\mathbf{S}}^W$ matrices in decreasing order and $\lbrace\breve{\mathbf{u}}_i\rbrace_{i=1}^L$, $\lbrace\breve{\mathbf{u}}_i^N\rbrace_{i=1}^L$ and $\lbrace\breve{\mathbf{u}}_i^W\rbrace_{i=1}^L$
denote the associated eigenvectors respectively.
\begin{enumerate}[label=(\alph*)]
	\item\label{item.a}
	For the narrowband part of the matrix $\breve{\mathbf{S}}$, we have $\sigma_1(\breve{\mathbf{S}}^N)=1$ and $\sigma_i(\breve{\mathbf{S}}^N)=0$ for $i=2,\cdots,L$ and $\breve{\mathbf{u}}_1^N = \mathbf{g}(f,\theta)/\sqrt{L}$.
	\item\label{item.b}
	The eigenvectors of matrix $\breve{\mathbf{S}}$ are the Hadamard product of the narrowband part into the wideband part as $\breve{\mathbf{u}}_i = \mathbf{g} \circ \breve{\mathbf{u}}^W_i$ for $i=1,\cdots,L$.
	\item\label{item.c}
	The $\breve{\mathbf{S}}$ and $\breve{\mathbf{S}}^W$ matrices have identical eigenvalues, $\lbrace \sigma_i(\breve{\mathbf{S}})\rbrace_{i=1}^L = \lbrace \sigma_i(\breve{\mathbf{S}}^W)\rbrace_{i=1}^L$ for $i=1,\cdots,L$.
	\item\label{item.d}
	The SS-Transform of the wideband part of the GSV has the maximum average power in $f\in[-\frac{B}{2},\frac{+B}{2}]$ for all $\mathbf{x}\in\mathbb{C}^L$ and $\|\mathbf{x}\|_2 = 1$.
\end{enumerate}
\end{theorem}
\begin{proof}
	See~\ref{app.0}.
\end{proof}
%
\subsection{Asymptotic Behavior of matrix $\breve{\mathbf{S}}$ } \label{sec.asymptotic}
Denoting $\lbrace \sigma_i(\mathbf{A}) \rbrace_{i=1}^{L}$ as the set of non-increasingly ordered eigenvalues of $\mathbf{A}$, we have,
\begin{align}
& \text{tr}(\breve{\mathbf{S}}) = \sum_{i=1}^{L} \breve{s}_{ii} = \sum_{i=1}^{L} \sigma_i(\breve{\mathbf{S}}) 
\end{align}
Moreover, matrices $\breve{\mathbf{S}}^N$, $\breve{\mathbf{S}}^W$ and $\breve{\mathbf{S}}$ are all positive semi-definite (see \ref{app.0}) , so we have,
\begin{equation}
\sigma_1(\breve{\mathbf{S}})\ge \sigma_2(\breve{\mathbf{S}})\ge \cdots \ge \sigma_L(\breve{\mathbf{S}})\ge 0
\end{equation}
In the narrowband case where $B \to 0$, the $\breve{\mathbf{S}}^W \to \mathbf{1}\mathbf{1}^T$, where $\mathbf{1} $ is a vector with all elements equal to one. 
Therefore, $\lim_{B \to 0}\breve{\mathbf{u}}_1(\theta,f_c,B)=\mathbf{g}(\theta,f_c)$.
Letting $m=1$ in \eqref{equ.hk}, we can write,
\begin{align}
&\lim_{B \to 0} \breve{s}_{k,l} = e^{j2\pi f_c (h_k - h_l)} = e^{j2\pi \frac{d}{\lambda}\sin\theta(k-l)} \nonumber\\
&\Rightarrow \lim_{m=1 , B \to 0}\breve{\mathbf{S}}(\theta,f_c,B) = \mathbf{a}(\theta,f_c)\mathbf{a}^H(\theta,f_c) \label{equ.S_narrow}
\end{align}
This reveals the motivation behind the naming of $\breve{\mathbf{u}}_1$ as the \textit{generalized steering vector}. 
Finally, in narrowband case all energy is concentrated on the first eigenvector,
\begin{equation} \label{equ.conv_narrow}
\lim_{B \to 0} \sigma_i(\breve{\mathbf{S}}) = \left\lbrace \begin{array}{ll}
L & i=1 \\
0 & i=2,\cdots,L
\end{array}\right.
\end{equation}

On the other hand for the case when bandwidth is very large, i.e. $B \to \infty$  , the $\breve{\mathbf{S}}^W$ matrix tends to a constant block diagonal matrix. Let the sampling frequency $f_s = \nu B$, then,
\begin{align}
&\lim_{B \to \infty} \breve{\mathbf{S}}^W(\theta,B) = \left[\begin{array}{ccc}
\breve{\mathbf{S}}^\infty & & \mathbf{0} \\
& \ddots &  \\
\mathbf{0} 		  & & \breve{\mathbf{S}}^\infty
\end{array}\right] \\
& \breve{s}^\infty_{k,l} = 
\frac{\sin\left(\pi \frac{|k-l|}{\nu} \right)}{\pi \frac{|k-l|}{\nu}} \; , \quad k,l\in\lbrace1,\cdots,m\rbrace
\end{align}
Recall that the eigenvalues of a block diagonal matrix are equal to the union of its diagonal submatrices eigenvalues:
\begin{equation}
\lbrace \sigma_k(\mathbf{A})\rbrace_{k=1}^{nm} = \bigcup_{i=1}^{n} \lbrace \sigma_k(\mathbf{A_{ii}})\rbrace_{k=1}^{m}
\end{equation}
Therefore, the eigenvalues of $\breve{\mathbf{S}}^W$ converge to a permutation of the $\lbrace\sigma_i(\breve{\mathbf{S}}^\infty)\rbrace_{i=1}^m$ as $B\to\infty$. 
\begin{equation} \label{equ.sigma_converge}
\lim_{B \to \infty} \sigma_i(\breve{\mathbf{S}}) = \left\lbrace \begin{array}{ll}
\sigma_1(\breve{\mathbf{S}}^\infty) & i=1,\cdots,N_S \\
\sigma_2(\breve{\mathbf{S}}^\infty) & i=N_S+1,\cdots,2N_S \\
\vdots & \\
\sigma_m(\breve{\mathbf{S}}^\infty) & i=L-N_S+1,\cdots,L \\
\end{array}\right.
\end{equation}
Accordingly, as the bandwidth increases, the dominant eigenvalue, corresponding to the GSV, decreases and the rest of the eigenvalues increase in such a way that all tend to form a smooth spectrum. 
An insightful example is depicted in Fig.~\ref{Fig.B_to_inf}, where $f_0$ is the maximum spatial-aliasing free frequency of the array. It shows the convergence discussed in \eqref{equ.conv_narrow} and \eqref{equ.sigma_converge}.
\begin{figure}
	\centering
	\includegraphics[width=0.64\linewidth]{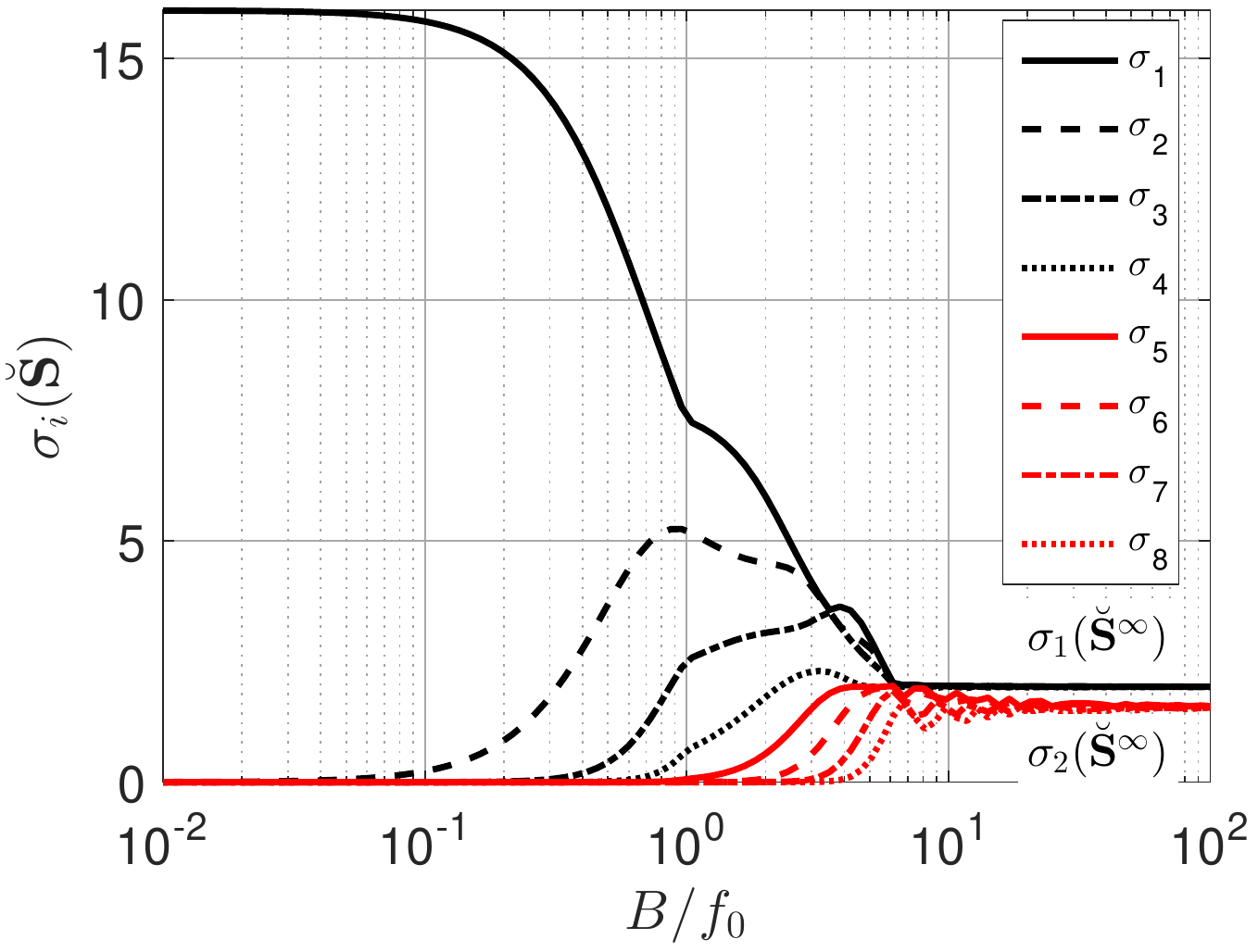}
	\caption{Asymptotic behavior of $\sigma_i(\breve{\mathbf{S}})$ for a ULA configuration and with spacing $d=c/2f_0$ , $m=4$, $N_S=4$ , $L=16$ , $\theta=40^\circ$ and $\nu=2$.}
	\label{Fig.B_to_inf}
\end{figure}
\subsection{Algorithms} \label{sec.algorithms}
In this part, two DOA estimation algorithms using the concept of modal orthogonality are proposed. Modal components of a wideband source at direction $\theta$ is attributed to the eigenvectors of $\breve{\mathbf{S}}(\theta)$ corresponding to non-zero eigenvalues. Sources' DOAs are locations that the modal components are orthogonal to the noise subspace $\mathbf{U}_n$.
We refer to this solution as wideband modal orthogonality (WIMO).
Naturally, $\mathbf{U}_n$ is substituted with $\hat{\mathbf{U}}_n$ in the noisy observation case.
Pure-WIMO\footnote{
	\color{MyBlue}
	The choice of '\textit{Pure}' prefix was due to the use of whole $\breve{\mathbf{S}}$ matrix in the orthogonality test without its modal decomposition.
	\color{black}} (p-WIMO) utilizes the whole $\breve{\mathbf{S}}$ matrix as the signal subspace. The spatial spectrum of p-WIMO is formulated as follows,
\begin{equation}\label{equ.pure-wico}
P_\text{p-WIMO}(\theta) = \frac{1}{\text{tr}\left(\hat{\mathbf{U}}_n^H \breve{\mathbf{S}}(\theta) \hat{\mathbf{U}}_n\right)}
\end{equation}
The denominator of p-WIMO can be rewritten as,
\begin{equation}\label{equ.wico-denom}
\text{tr}\left(\hat{\mathbf{U}}_n^H \breve{\mathbf{S}}(\theta) \hat{\mathbf{U}}_n\right) = 
\sum_{l=1}^{L} \sigma_l(\breve{\mathbf{S}}) \sum_{k=P+1}^{L} \left\| \hat{\mathbf{u}}^H_{n,k} \breve{\mathbf{u}}_l \right\|^2
\end{equation}
where $\hat{\mathbf{u}}_{n,k}$ denotes the $k$th column of $\hat{\mathbf{U}}_n$. This expression shows that each modal component contribution is proportional to its corresponding eigenvalue, i.e., weak eigenvectors will have a negligible effect on the nullity of denominator.
The second algorithm, named 1-WIMO\footnote{\color{MyBlue} '\textit{1}' prefix denotes using the first modal component in the orthogonality test. \color{black}}, tests only the GSV (defined in Def.~\ref{def.gsv}) orthogonality to the noise subspace. The spatial spectrum of the 1-WIMO can be written as,
\begin{equation} \label{equ.1-wimo}
P_\text{1-WIMO}(\theta) = \frac{1}{\breve{\mathbf{u}}_1(\theta)^H \hat{\mathbf{U}}_n\hat{\mathbf{U}}_n^H \breve{\mathbf{u}}_1(\theta)}
\end{equation}
According to the asymptotic behavior of $\breve{\mathbf{S}}$ in narrowband situation (see section~\ref{sec.asymptotic}), 1-WIMO tends toward the celebrated MUSIC algorithm as $B\to 0$. 
It is noteworthy that both $\breve{\mathbf{S}}(\theta)$ and $\breve{\mathbf{u}}_1(\theta)$ computations can be performed offline for arbitrary DOAs.
Moreover, WIMO requires one EVD for the whole band and no subband processing is involved. These properties make WIMO a practically attractive approach with tractable computational complexity for real-time applications. 
A step by step description of the proposed algorithms are expressed in Algorithm~\ref{alg.WIMO}.
\begin{algorithm} [!t]
	\caption{p-WIMO and 1-WIMO wideband DOA estimators.}
	\label{alg.WIMO}
	\begin{algorithmic}[1]
\Statex \Comment Doing offline calculations
\Procedure{WimoOffline}{$B$,$f_c$,$\mathbf{z}$,$m$,$\Theta$} 

\State Compute $\breve{\mathbf{S}}(\theta,B,f_c)$ for $\theta\in\Theta$ points using \eqref{equ.smat} (in non-uniform PSD use \eqref{equ.non-uniform} instead).

\State Extract $\breve{\mathbf{u}}_1(\theta)$ as the eigenvector of $\breve{\mathbf{S}}(\theta,B,f_c)$ corresponding to largest eigenvalue.
\EndProcedure

\Comment Doing online calculations
\Procedure{WimoOnline}{$\mathbf{y}(t)$,$m$,$\breve{\mathbf{S}}(\Theta)$,$\breve{\mathbf{u}}_1(\Theta)$,
	alg} 
\State Form $\tilde{\mathbf{y}}(t)$ using \eqref{equ.stom}.
\State Estimate $\hat{\mathbf{S}}_{\tilde{\mathbf{y}} \tilde{\mathbf{y}}}$ using \eqref{equ.Syy_hat}.

\State Compute $\hat{\mathbf{U}}_n$ using \eqref{equ.Syy_EVD}.

\If {alg = p-WIMO}
\State Compute $P(\theta)$ using \eqref{equ.pure-wico}.

\ElsIf{alg = 1-WIMO}
\State Compute $P(\theta)$ using \eqref{equ.1-wimo}.
\EndIf

\EndProcedure
\end{algorithmic}
\end{algorithm}
\vspace*{-5pt}
\subsection{WIMO Parameter Selection} \label{sec.parameter}
The two parameters $m$ (length of temporally delayed samples) and $P$ (number of eigenvalues corresponding to the signal subspace) are required in \eqref{equ.stom} and \eqref{equ.Syy_EVD} respectively.
Since the signal subspace of broadband source does not lie in a 1-dimensional subspace, and the value of $m$  directly determines the observation space dimension, $m$ should be chosen large enough so that the following constraint always holds,
\begin{equation}\label{equ.m_const1}
\varepsilon^{(\mathcal{K})} < mN_S
\end{equation}
where $\varepsilon^{(\mathcal{K})}$ is the effective signal subspace dimension in the presence of $\mathcal{K}$ sources.
In \cite{buckley1988broad}, an upper bound is introduced for $\varepsilon^{(1)}$ and $\varepsilon^{(\mathcal{K})}$ as,
\begin{align}
\hat{\varepsilon}_\text{BASS-ALE}^{(1)} &= m + N_S \\
\hat{\varepsilon}^{(\mathcal{K})}_\text{BASS-ALE} &= \mathcal{K} \cdot \hat{\varepsilon}_\text{BASS-ALE}^{(1)} \label{equ.eps_bass_ale}
\end{align}
We propose using $\breve{\mathbf{S}}$, to approximate $\varepsilon^{(\mathcal{K})}$ as follows,
\begin{equation} \label{equ.eps_hat}
\hat{\varepsilon}^{(\mathcal{K})} = \text{rank}\left(
\sum_{k=1}^{\mathcal{K}}
\breve{\mathbf{S}}(\theta_k , f_{c} , B)\right)
\end{equation}
\eqref{equ.eps_hat} in the narrowband case obviously equals to $\mathcal{K}$.
To account for the maximum probable time-bandwidth product in \eqref{equ.eps_hat}, 
\color{black}
it is enough to set $\theta$=$90^\circ$ in \eqref{equ.eps_hat} as,
\begin{equation} \label{equ.eps_hat_max}
\hat{\varepsilon}^{(\mathcal{K})}_\text{max} = \mathcal{K} \cdot \text{rank}\left(
\breve{\mathbf{S}}(90^\circ , f_{c} , B)\right)
\end{equation}
Empirical results show that the proposed algorithms are not sensitive to $m$, and $m\ge 2$ subject to the aforementioned constraint \eqref{equ.m_const1} leads to similar results.

In the selection of $P$ parameter, two approaches can be considered. The first one is to set $P=\hat{\varepsilon}^{(\mathcal{K})}$. This solution achieves best performance since it incorporates complete sources information in the signal subspace estimation. But, this approach is infeasible in real applications due to the unknown sources' direction of arrivals. The second solution is to exploit well-known order estimation methods such as minimum description length (MDL) criterion \cite{Kailatha1984}. This approach has a straightforward implementation but always underestimates signal subspace dimension due to small eigenvalues which are smaller than the noise floor. 
The performance degradation due to underestimation of $P$ is larger in p-WIMO. Since p-WIMO incorporates the whole approximated signal subspace, some weaker eigenvectors of $\hat{\mathbf{S}}_{\tilde{\mathbf{y}} \tilde{\mathbf{y}}}$ remain in $\hat{\mathbf{U}}_n$, and then the orthogonality test between $\hat{\mathbf{U}}_n$ and $\breve{\mathbf{S}}(\theta)$ adds up some small non-orthogonal terms from signal subspace. Therefore, to achieve a better performance for WIMO (especially in ultra-wideband scenarios) a manual increment of $P$ relative to the $P_\text{MDL}$ is recommended. The following selection rules are derived empirically,
\begin{enumerate}
	\item
	$P > \max\{P_\text{MDL} , \hat{\varepsilon}^{(1)}_\text{max}\}$.
	\item 
	$0.2 \le P/L \le 0.6$ for 1-WIMO and $0.5 \le P/L \le 0.7$ for p-WIMO.
\end{enumerate}
\subsection{WIMO for non-uniform spectrum} \label{sec.non-uniform}
In many applications such as active radars and sonars and also broadband communications, the signal power spectrum is known but is not uniformly distributed in the bandwidth. 
Moreover, in non-cooperative scenarios, the sources' spectrum may be estimated from the observations prior to DOA estimation.
In these cases, the simplifying assumption on uniform power spectral density in \eqref{equ.temp1} is not accurate and would cause a mismatch loss in DOA estimation. To extend WIMO in such situations, it is sufficient to compute $\breve{\mathbf{S}}(\theta,f_c,B)$ from the generic equation \eqref{equ.int_approx} without uniform spectrum assumption. Then, the approximation matrix entries are obtained through the following numerical integration,
\begin{equation} \label{equ.non-uniform}
\breve{s}_{k,l} = \int_{f_l}^{f_h} S(f) e^{j2\pi (h_k - h_l)f} \cdot df
\end{equation}
where $S(f)$ can be replaced with estimated PSD in non-cooperative case. Also, if any information about source's correlation function is available, equation \eqref{equ.corr-function} can be applied directly.
Finally, 
except the line~2 of Algorithm~\ref{alg.WIMO}, other parts of the WIMO algorithm remain unchanged for uniform and nonuniform PSD.
\color{black}
Note that (as in the uniform spectrum case) numerical integration in \eqref{equ.non-uniform} can be computed offline.

\subsection{WIMO computational complexity} \label{sec.complexity}
Table~\ref{table.complexity} details the computational complexity of the proposed algorithms. 
It is shown that WIMO computational complexity is approximately of order $\mathcal{O}(m^3 N_S^3)$, then excessive increase of $m$ (beyond the constraint in \eqref{equ.m_const1}) can make WIMO runtime larger than those methods involving subband processing. Moreover, in comparison to p-WIMO, 1-WIMO has less complexity.
\begin{table}[!t]
	\scriptsize
	\centering
	\begin{threeparttable}[b]
		\caption{WIMO computational complexity analysis.}
		\label{table.complexity}
		\begin{tabular}{p{0.37\linewidth}p{0.12\linewidth}p{0.35\linewidth}}
\toprule[1.0pt]
Operation & Output & Complexity order \tabularnewline
\midrule
STCM approximation \eqref{equ.Syy_hat}	& $\hat{\mathbf{S}}_{\tilde{\mathbf{y}} \tilde{\mathbf{y}}}$
& $\mathcal{O}((M-m+1)L^2)$ \tabularnewline
Noise subspace estimation \eqref{equ.Syy_EVD} & $\hat{\mathbf{U}}_n$ 
& $\mathcal{O}(L^3 + (L\log^2L)\log b)$\tnote{$\dag$} \tabularnewline
p-WIMO pseudo-spectrum & $P_\text{p-WIMO}(\theta)$ 
& $\mathcal{O}( (L-P)L^2 + L(L-P)^2 )$ \tabularnewline
1-WIMO pseudo-spectrum & $P_\text{1-WIMO}(\theta)$ 
& $\mathcal{O}( (L+1)(L-P) )$ \tabularnewline
\bottomrule[1.0pt]	
\end{tabular}

\begin{tablenotes}
\item[$\dag$] $b=64$ for double precision. See \cite{Pan1999} for details.
\end{tablenotes}
	\end{threeparttable}
\end{table}
WIMO's main computational cost is related to the eigenvalue decomposition of STCM.
The main advantage of WIMO over some other methods is that the signal subspace approximation is done offline and noise subspace estimation is obtained through a single STCM decomposition for the whole band and all DOAs. 

A numerical example comparing the runtime of the different methods in a definite scenario is presented in the next section. Moreover, since the main computational load in subspace-based schemes belongs to singular value decomposition (SVD), a comparison of number of required SVD among these methods can be insightful. Recall that the computational complexity of $\mathtt{svd}(\mathbf{A}_{l\times p})$ is $\mathcal{O}(4l^2 p + 22p^3)$ \cite{Golub1996}. Table~\ref{table.svd} details this comparison, in which $N_{bin}$ represents the number of frequency bins computed as $N_{bin}=\left[\frac{B}{f_s} N_{FFT}\right]$, $N_{FFT}$ is the number fast Fourier transform (FFT) points, $\left[\cdot\right]$ is rounding operator and $N_\theta$ is the number DOAs in which spatial spectrum is calculated.
\begin{table}[!t]
	\scriptsize
	\centering
	\begin{threeparttable}[b]
		\caption{Comparison of some of the subspace-based wideband DOA estimation methods in term of the number of required singular value decomposition.}
		\label{table.svd}
		\begin{tabular}{p{0.22\linewidth}p{0.4\linewidth}}
\toprule[1.0pt]
Method & Number of SVD \tabularnewline
\midrule
1-WIMO	& $\mathtt{svd}(\mathbf{A}_{mN_S\times mN_S})$ \tabularnewline
p-WIMO	& $\mathtt{svd}(\mathbf{A}_{mN_S\times mN_S})$ \tabularnewline
BASS-ALE \cite{buckley1988broad}	& $\mathtt{svd}(\mathbf{A}_{mN_S\times mN_S})$ \tabularnewline
CSSM\tnote{$\dag$} \cite{Wang1985}	& $(N_{bin}+1)\mathtt{svd}(\mathbf{A}_{N_S\times N_S})$ \tabularnewline
IMUSIC \cite{Kailatha1984}	& $N_{bin}\mathtt{svd}(\mathbf{A}_{N_S\times N_S})$ \tabularnewline
TOPS \cite{Yoon2006a}	& $N_{bin}\mathtt{svd}(\mathbf{A}_{N_S\times N_S}) + 
		  N_\theta\mathtt{svd}(\mathbf{A}_{\mathcal{K}\times (N_{bin}-1)(N_S-\mathcal{K})})$ \tabularnewline
Squ-TOPS \cite{Okane2010}& $N_{bin}\mathtt{svd}(\mathbf{A}_{N_S\times N_S}) + 
N_\theta\mathtt{svd}(\mathbf{A}_{\mathcal{K}\times (N_{bin}-1)\mathcal{K}})$ \tabularnewline
STEP\tnote{*} \cite{Yin2018b}	& $N_{\theta}\mathtt{svd}(\mathbf{A}_{mN_S\times mN_S})$ \tabularnewline
\bottomrule[1.0pt]	
\end{tabular}

\begin{tablenotes}
	\item[$\dag$] Assuming the number of source's and their DOAs' pre-estimates are available.
	\item[*] With given number of sources.
\end{tablenotes}
	\end{threeparttable}
\end{table}
\color{black}
\vspace*{-5pt}

\section{Simulation Results} \label{sec.simulation}
In this section, the simulation results for the proposed wideband DOA estimators and IMUSIC \cite{Kailatha1984}, CSSM \cite{Wang1985}, TOPS \cite{Yoon2006a}, Squared-TOPS \cite{Okane2010}, BASS-ALE \cite{buckley1988broad}, STEP \cite{Yin2018b}, 
\color{MyBlue} spatial-only \cite{Agrawal2000} \color{black}
and also sparse wideband DOA estimators  $\ell_1$-SVD \cite{malioutov2005sparse}, W-CMSR \cite{liu2011direction} and, W-LASSO \cite{Hu2012a} are presented.
A ULA with 8 omnidirectional sensors and element spacing $d=\lambda_{min}/2$, is considered. The wave propagation velocity $c$, is equal to 1500$(m/s)$ which corresponds to propagation speed of underwater acoustic. Wideband sources are simulated as a colored Gaussian processes. Two metrics are utilized to measure the frequency broadness of a signal; bandwidth ratio denoted by $\eta$ and bandwidth scale expressed by $\gamma$. For a given wideband signal with frequency content $f\in [f_l, f_h]$:
\begin{equation}\label{equ.BF}
\eta \triangleq 2\left(\frac{f_h - f_l}{f_h + f_l}\right) \quad , \quad
\gamma \triangleq \frac{f_h}{f_l}
\end{equation}
obviously $0\le\eta\le 2$ and $\gamma \ge 1$. 

In the first numerical example, the eigenvalues of approximation matrix $\breve{\mathbf{S}}$, introduced in \eqref{equ.smat}, is compared with the STCM $\hat{\mathbf{S}}_{\tilde{\mathbf{y}} \tilde{\mathbf{y}}}$. In this simulation, $f_l$=1.5KHz and $f_h$=4.5KHz which corresponds to $\eta$=100\%. Sampling frequency $f_s$=10KHz, $m$=5 and $\sigma_n^2$=0, until all signal eigenvalues become observable. The results are illustrated in Fig.~\ref{Fig.EigenValues_Comparison}. It is shown that for a single and two wideband sources, $\lbrace \sigma_k(\breve{\mathbf{S}})\rbrace_{k=1}^L$ are in a close match with $\lbrace \sigma_k(\hat{\mathbf{S}}_{\tilde{\mathbf{y}} \tilde{\mathbf{y}}})\rbrace_{k=1}^L$. The description of the two scenarios are given in the figure's caption.
\begin{figure}
	\centering
	\includegraphics[width=0.64\linewidth]{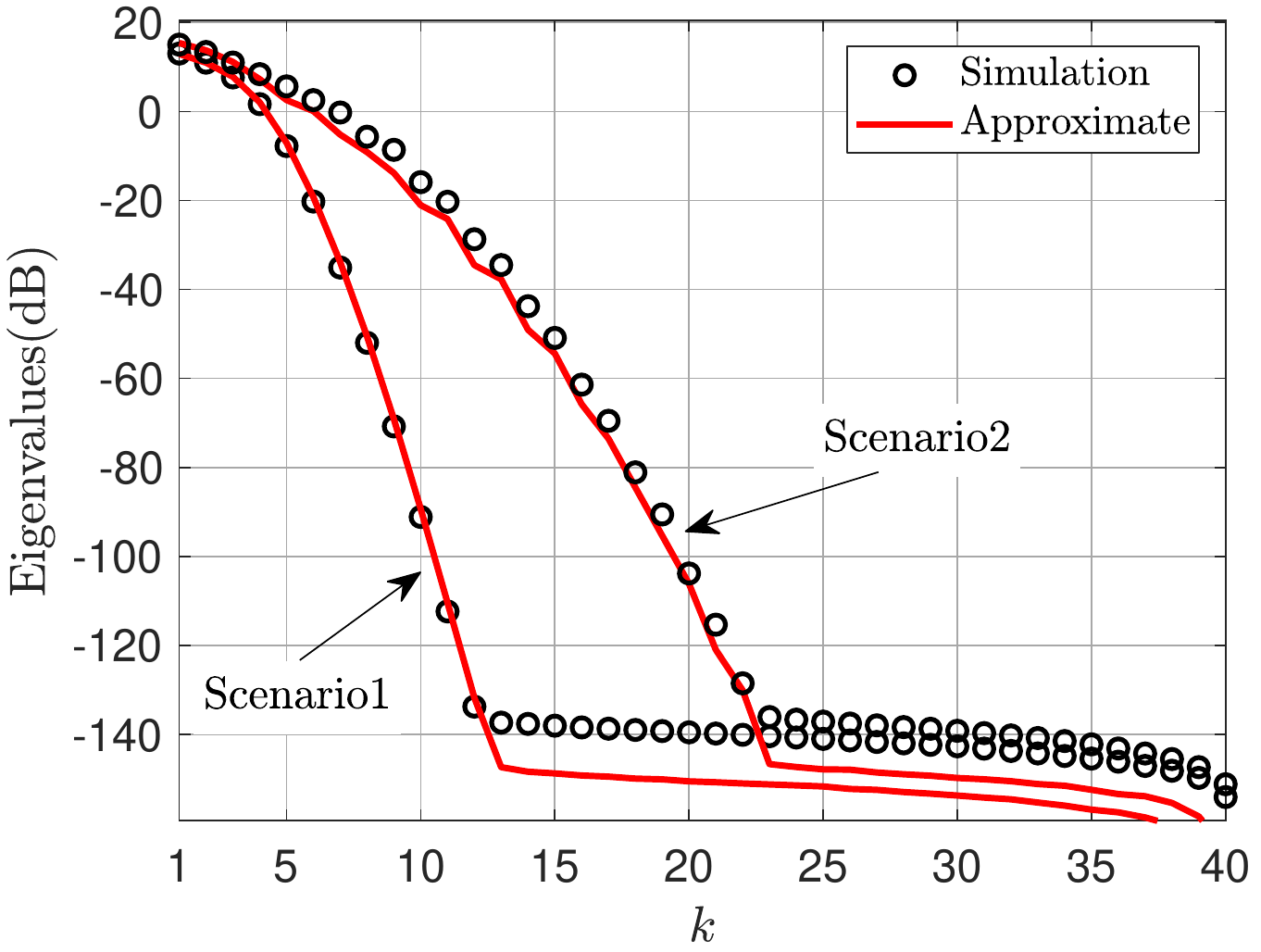}
	\caption{Comparison between eigenvalues of $\breve{\mathbf{S}}$ and $\hat{\mathbf{S}}_{\tilde{\mathbf{y}} \tilde{\mathbf{y}}}$. In scenario-1 single source is located at $\theta=40^\circ$ and in scenario-2 two equal power sources are at $\theta=[40^\circ,60^\circ]$. Eigenvalues of $\hat{\mathbf{S}}_{\tilde{\mathbf{y}} \tilde{\mathbf{y}}}$ are averaged over 100 runs.}
	\label{Fig.EigenValues_Comparison}
\end{figure}
In the second example, the effective signal subspace dimension $\varepsilon$, against signal bandwidth is investigated. The upper bound \eqref{equ.eps_bass_ale}, introduced in \cite{buckley1988broad}, our proposed upper bound in \eqref{equ.eps_hat_max} and, the estimate of effective dimension proposed in \eqref{equ.eps_hat} are compared with the true one. In this example, the true $\varepsilon$ is set to $\mathrm{rank}(\hat{\mathbf{S}}_{\tilde{\mathbf{y}} \tilde{\mathbf{y}}})$, sources DOAs are $40^\circ$ and $60^\circ$ and $f_h$=4.5KHz. The result is illustrated in Fig.~\ref{Fig.Rank_Comparison}. It is shown that the proposed $\hat{\varepsilon}^{(\mathcal{K})}$ leads to better approximation of the effective signal subspace dimension and also $\hat{\varepsilon}^{(\mathcal{K})}_\text{max}$ results in tighter upper bound compared to the constant upper bound proposed in \cite{buckley1988broad}.
\begin{figure}
	\centering
	\includegraphics[width=0.64\linewidth]{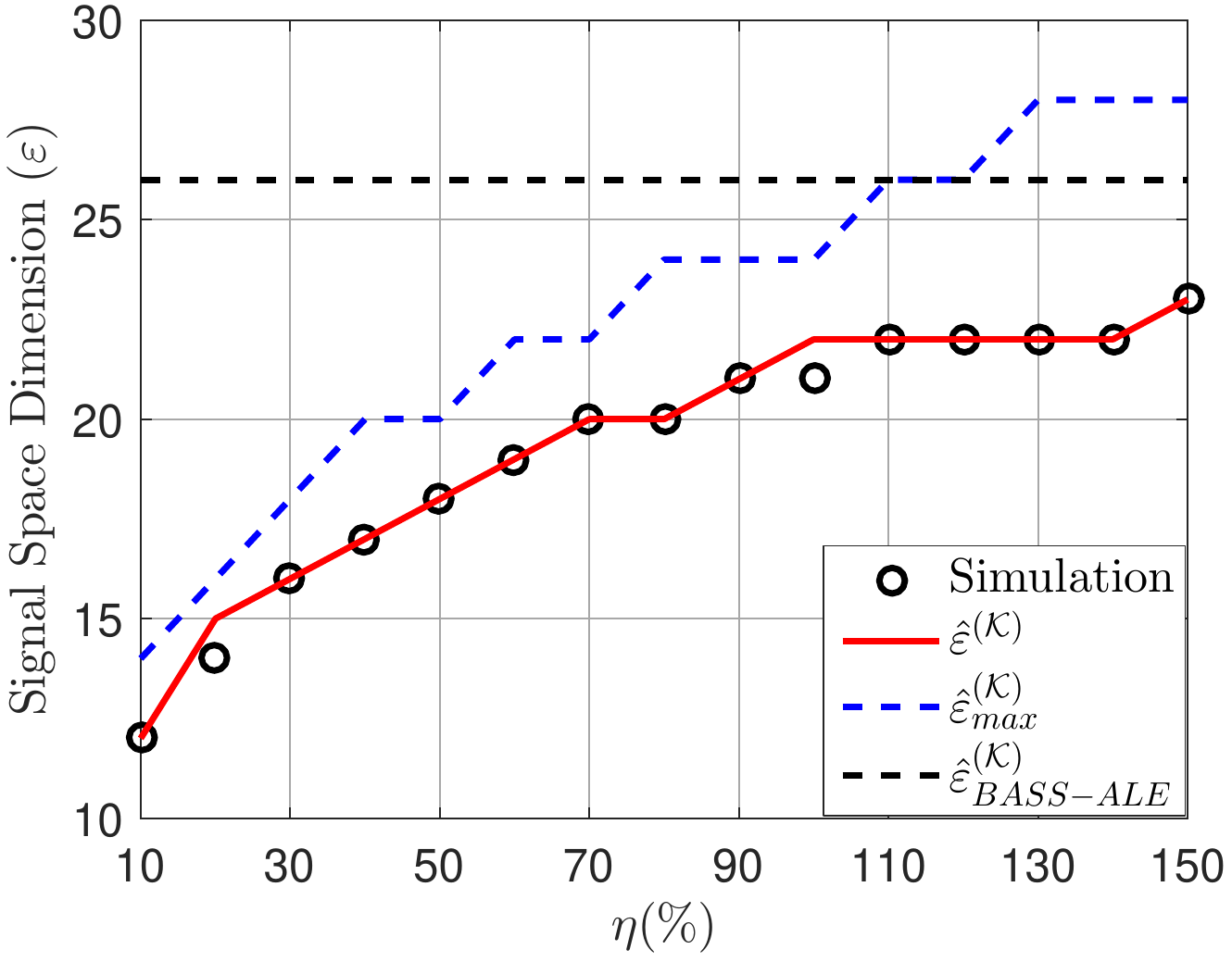}
	\caption{The effective broad-band signal subspace dimension $\varepsilon$ versus bandwidth ratio $\eta$.} 
	\label{Fig.Rank_Comparison}
\end{figure}
In the third numerical example, we consider the effect of temporal lag order $m$ in the 1-WIMO and p-WIMO on resulting RMSE, for three different bandwidth ratio 25\%, 50\% and, 100\%. $f_h$=4.5KHz, SNR=20dB , $m$=5 and $P$ is set to $P_\text{MDL}$ for both methods. The result is illustrated in Fig.~\ref{Fig.MSE_vs_m}. It is noteworthy that, the RMSE improvement for $m>1$ increases with bandwidth ratio. In other words, the effect of spatial-temporal observation model ($m>1$) compared to the simple array snapshot ($m$=1) becomes more sensible as the bandwidth ratio increases. As mentioned in Section~\ref{sec.parameter}, the proposed methods show little sensitivity for $m\ge 2$.
\begin{figure}
	\centering
	\subfloat[1-WIMO]{
		\includegraphics[width=0.45\linewidth]{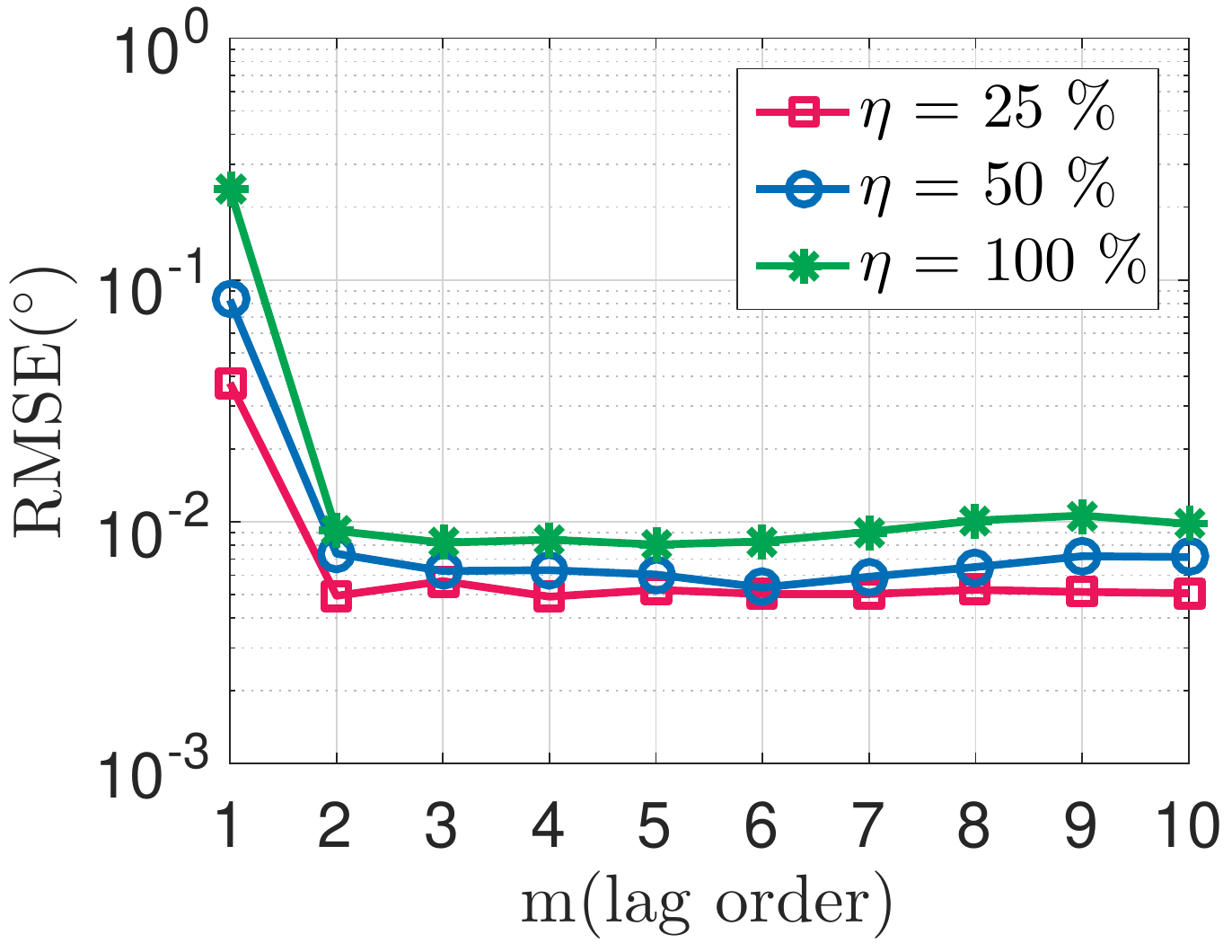}\label{subfig.1-wimo}}
	\subfloat[p-WIMO]{
		\includegraphics[width=0.45\linewidth]{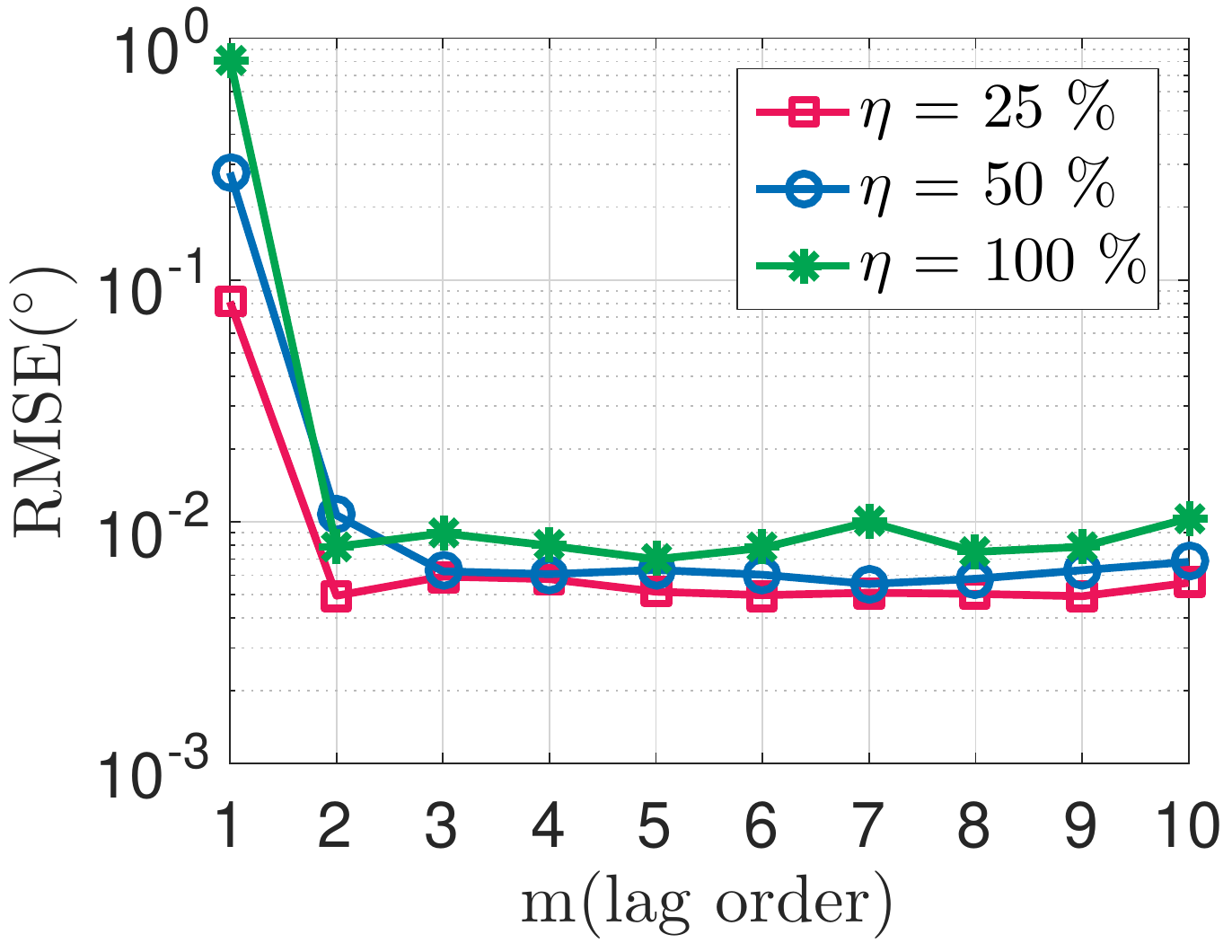}\label{subfig.p-wimo}}
	\caption{RMSE for 1-WIMO and p-WIMO versus $m$ values and different $\eta$. Each point is the average of 100 runs.}
	\label{Fig.MSE_vs_m}
\end{figure}

For the subsequent simulations, the pair $(m,P)$ is set to (6,15) for 1-WIMO and (9,45) for p-WIMO. Furthermore, in the CSSM method, true DOAs are supplied as the initial focusing angles and rotational subspace focusing matrix \cite{Kaveh1988} is applied. 
As the fourth example, we compare the spatial spectrum. Two wideband sources with $f_l$=1.5KHz and $f_h$=4.5KHz and DOAs $10^\circ$,$20^\circ$ with SNR=0dB are considered. The resulting spatial spectrum is illustrated in Fig.~\ref{Fig.DOA_Comparison}. 
The sampling frequency $f_s$=10kHz and there are 8192 samples (observation duration is 820 msec).
Some methods show a bias in the peaks locations.
The STEP method requires a critical parameter $M_\theta$ (see \cite{Yin2018b} for details) 
which is estimated through the procedure proposed in \cite{Yin2018b}. In STEP$^*$, $M_\theta$ is set to $\mathcal{K}+1$. 
The $^*$ superscript after the STEP title, emphasizes that this method is excluded from the unknown source's number assumption and $M_\theta$ is set to $\mathcal{K}+1$ for all further simulations.
\color{black}
A magnifier tool is utilized in Fig.~\ref{Fig.DOA_Comparison} which better shows the position of the spectrum's peaks. It is also observed that STEP$^*$ suffers from a bias in the source location estimation. This disadvantage becomes more evident in the subsequent simulations comparing the root mean squared error (RMSE).
\begin{figure}
	\centering
	\includegraphics[width=0.64\linewidth]{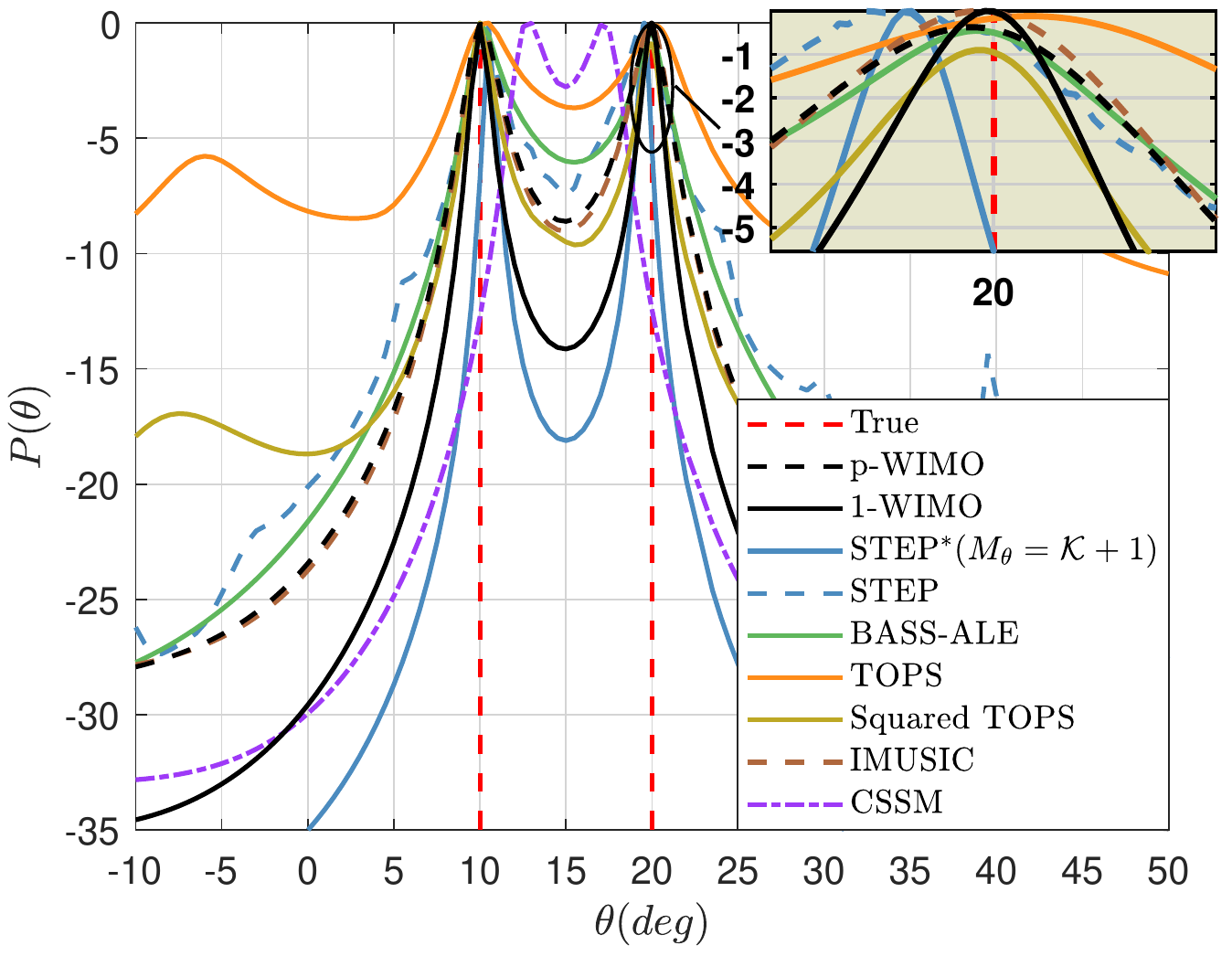}
	\caption{The averaged spatial spectra of different wideband DOA estimators. Each spectrum is averaged over 100 runs.}
	\label{Fig.DOA_Comparison}
\end{figure}

The probability of resolution versus SNR is compared in the next example. Two Gaussian distributed sources located at $15^\circ$ and $25^\circ$ is considered 
with number of snapshots the same as the previous example.
\color{black}
 Successful separation is attributed to the peaks with minimum prominence of 3dB and angle error smaller than 1.0$^\circ$.
In addition, spectra with the number of peaks greater or less than the true ones are considered as failed scenarios. The results for three bandwidth ratios of 40\%, 100\%, and 164\% are illustrated in Fig.~\ref{Fig.SeperationProb_vs_SNR}. 
The STEP$^*$ (with the assumption of known sources' number) and 1-WIMO have a superior probability of resolution in low SNR ultra-wideband scenarios.
p-WIMO method is ranked second in terms of probability of resolution for the three simulated bandwidth ratios.

In the sixth example, bandwidth dependent performance is inspected. The bandwidth scale $\gamma$, defined in \eqref{equ.BF}, is used as the broadness index of the input signal. Obviously, $\gamma$=1 is equivalent to the narrowband case. In each $\gamma$, $f_h$ is set to 4KHz and $f_l=f_h/\gamma$. The sources are located at $15^\circ$ and $25^\circ$ with SNR=0dB
 and M=8192.
\color{black}
 Fig.~\ref{Fig.Separation_Prob_vs_B} shows the simulation result. 
\begin{figure*}
	\centering
	\subfloat[$f_l$=3KHz , $f_h$=4.5KHz ($\gamma$=1.5 , $\eta$=40\%)]{
		\includegraphics[width=0.32\linewidth]{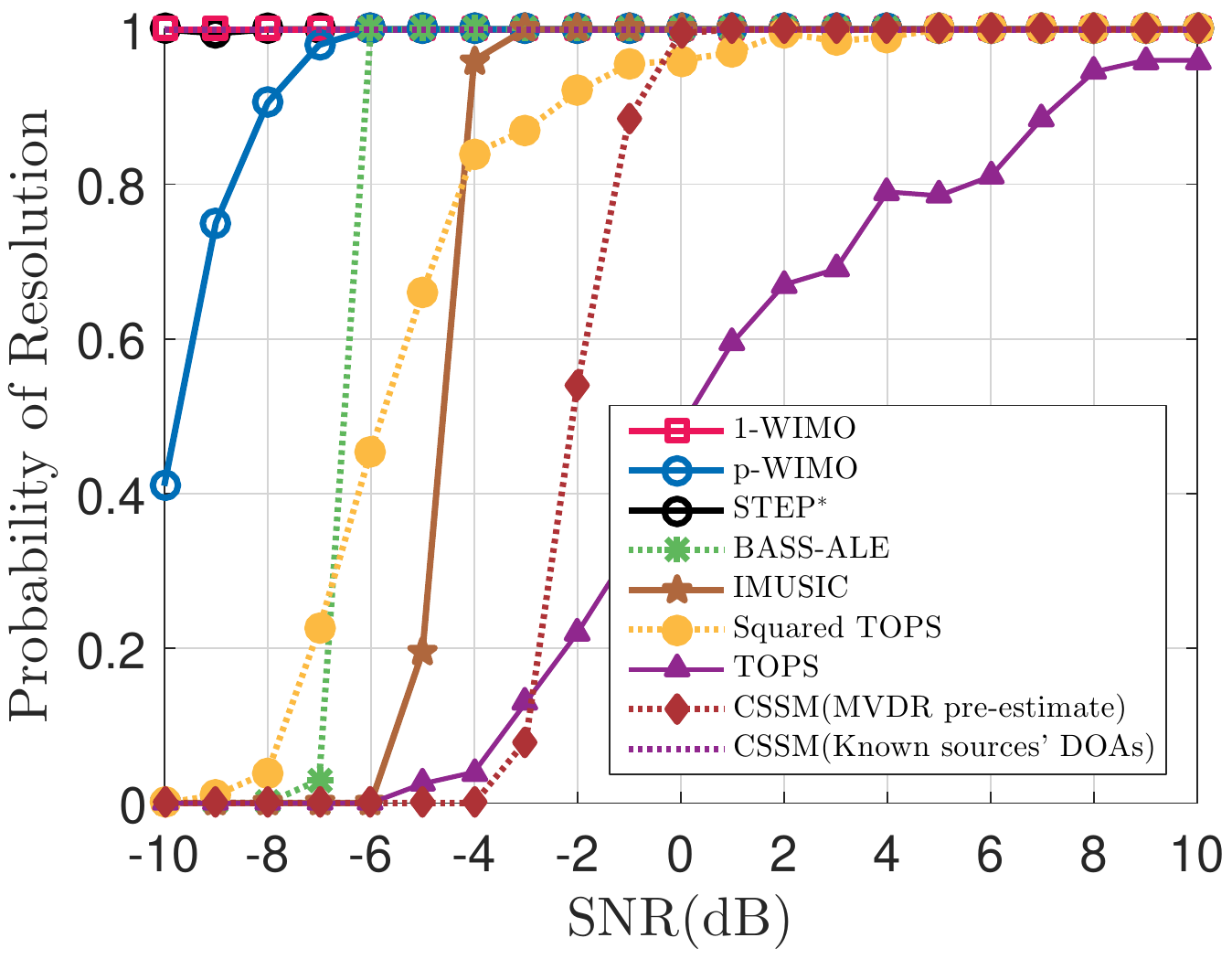}\label{subfig.bwr40}}
	\subfloat[$f_l$=1.5KHz , $f_h$=4.5KHz ($\gamma$=3 , $\eta$=100\%)]{
		\includegraphics[width=0.32\linewidth]{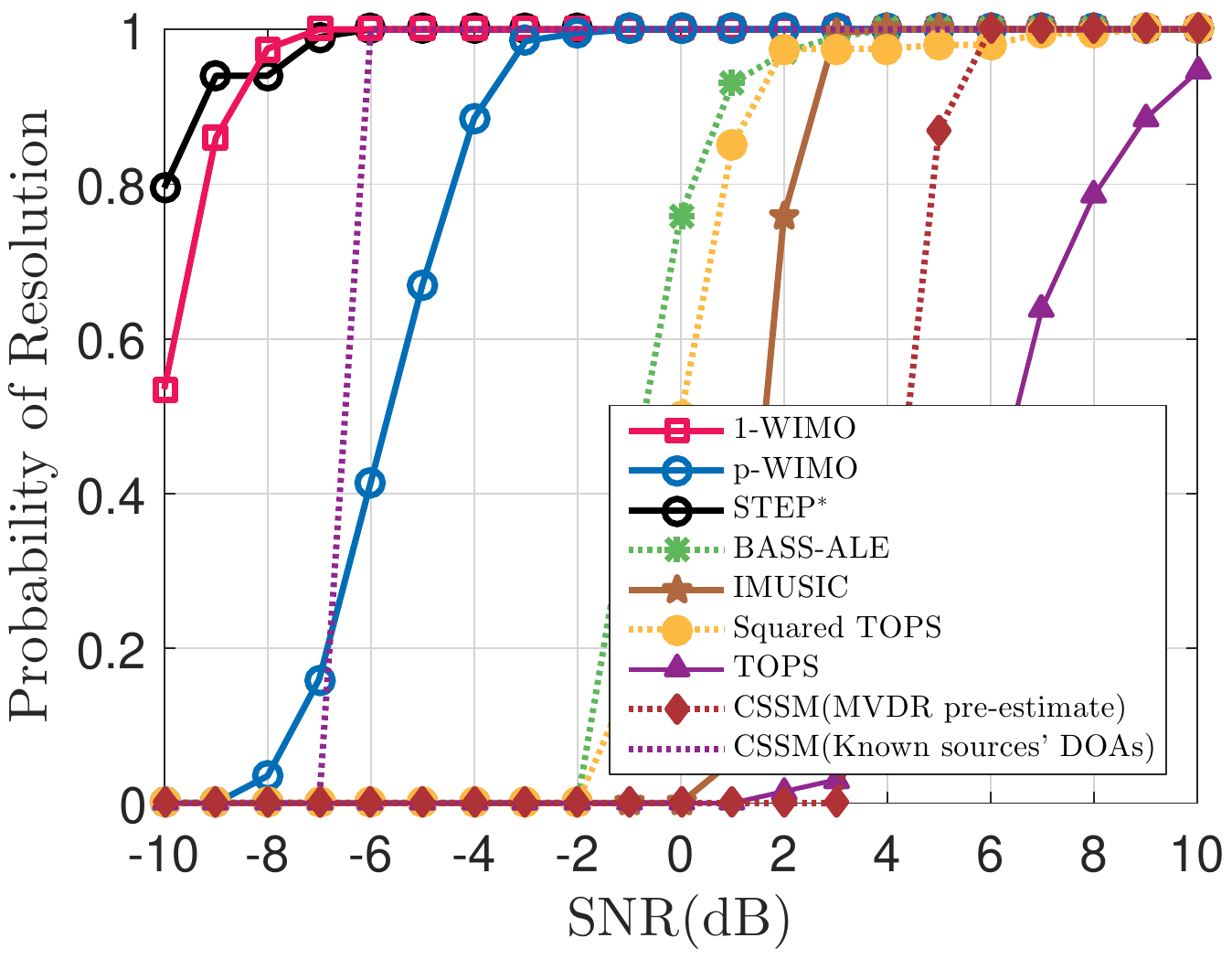}\label{subfig.bwr100}}
	\subfloat[$f_l$=0.4KHz , $f_h$=4.0KHz ($\gamma$=10 , $\eta$=164\%)]{
		\includegraphics[width=0.32\linewidth]{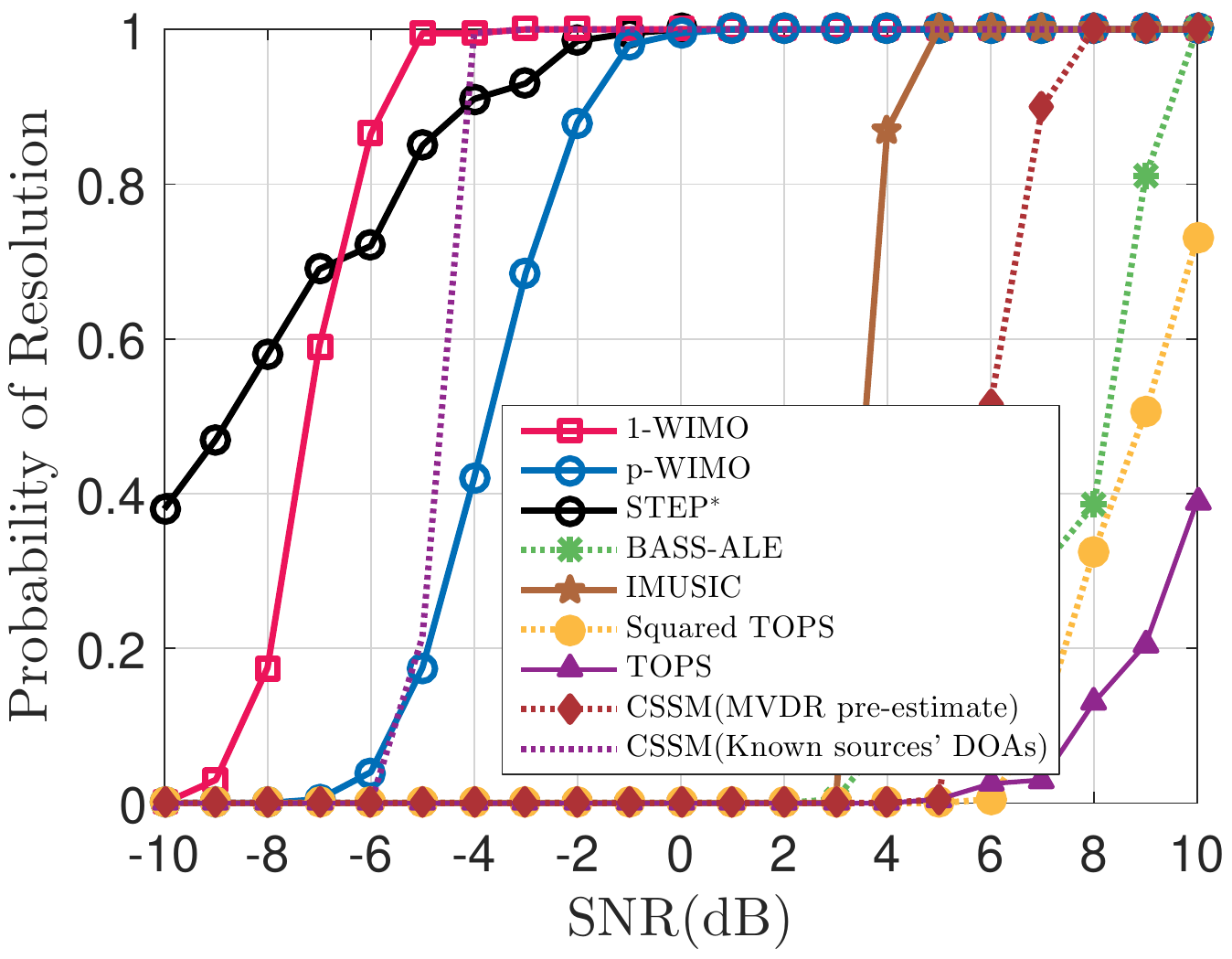}\label{subfig.bwr164}}
	\hfil
	\caption{Probability of resolution versus SNR for three bandwidth ratios 40\%, 100\% and, 164\%. Each point is the average of 200 runs.}
	\label{Fig.SeperationProb_vs_SNR}
\end{figure*}
As expected, bandwidth increase provides a challenging situation for wideband DOA estimators. 
p-WIMO, 1-WIMO and STEP$^*$ show superior performance in resolving the two sources, even when the sources occupy a broad frequency band, i.e. $\gamma$=100. 
\begin{figure}[!b]
	\centering
	\includegraphics[width=0.64\linewidth]{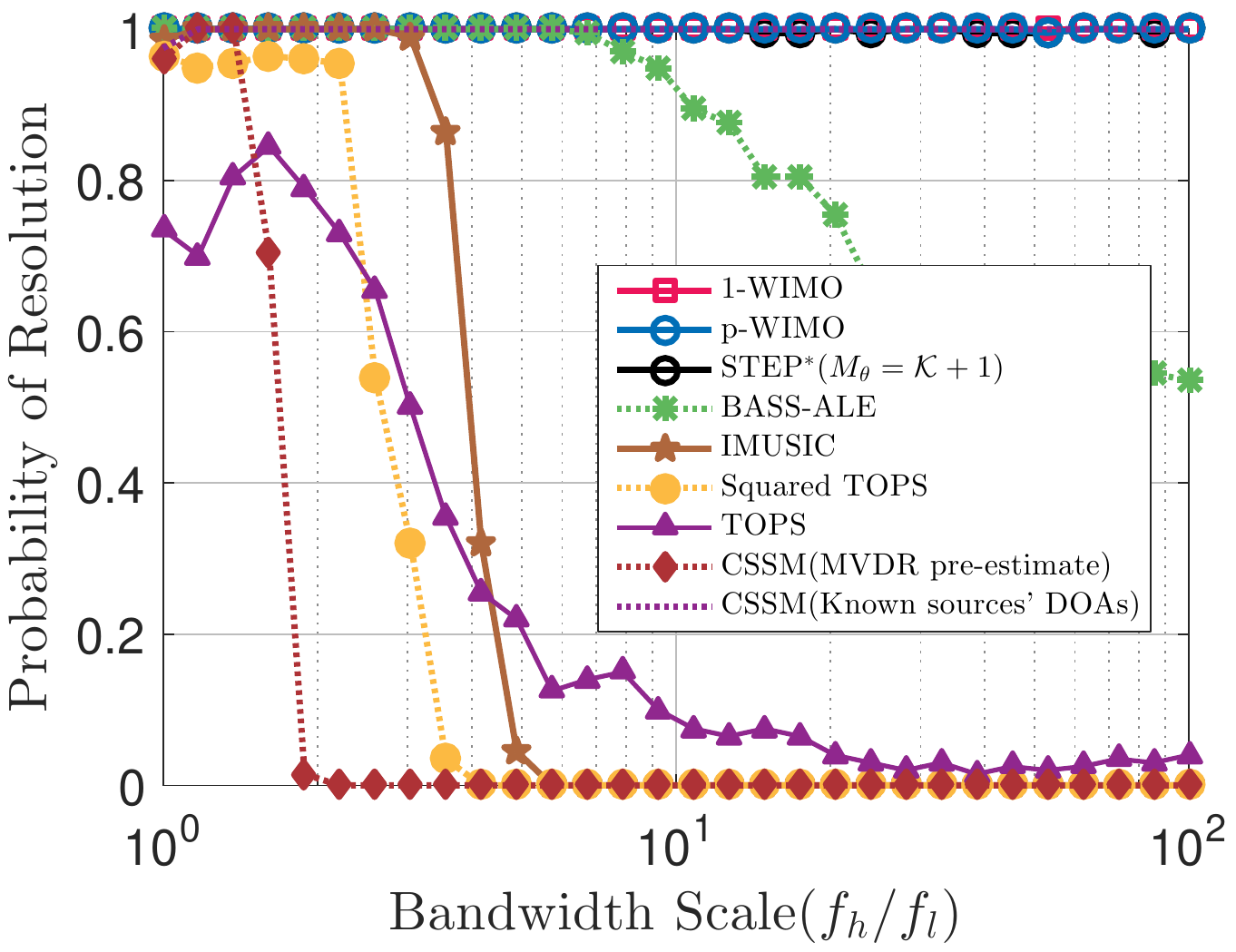}
	\caption{Separation probabilities versus bandwidth scale. Probability values are calculated through 200 runs.}
	\label{Fig.Separation_Prob_vs_B}
\end{figure}
RMSE comparison versus SNR is investigated in the seventh simulation. Two Gaussian distributed sources with $\eta$=100\% ($f_l$=1.5KHz , $f_h$=4.5KHz) at $-5^\circ+\nu$ and $+5^\circ+\nu$ are considered. $\nu$ is chosen randomly with a uniform distribution within $[-0.5^\circ , +0.5^\circ]$. Other simulation circumstances are the same as the previous ones. The results are illustrated in Fig.~\ref{Fig.MSE_vs_SNR}. 
The Cramer-Rao lower bound (CRLB) in the wideband case is also calculated by the procedure given in \cite{Wang1985}.
\color{black}
1-WIMO has the best RMSE performance at low SNR regime and also p-WIMO shows identical RMSE at medium and high SNRs. 
The CSSM (with known sources' DOAs) stands for CSSM using rotational signal-subspace focusing matrix \cite{Kaveh1988} with exact sources' locations as focusing angles. The CSSM performance using MVDR pre-estimates is also shown for comparison.
\color{black}
As already mentioned, STEP$^*$ (with the given number of sources) despite successful resolving capability, suffers from an SNR independent bias in DOA estimation. 
\begin{figure}[!b]
	\centering
	\includegraphics[width=3in]{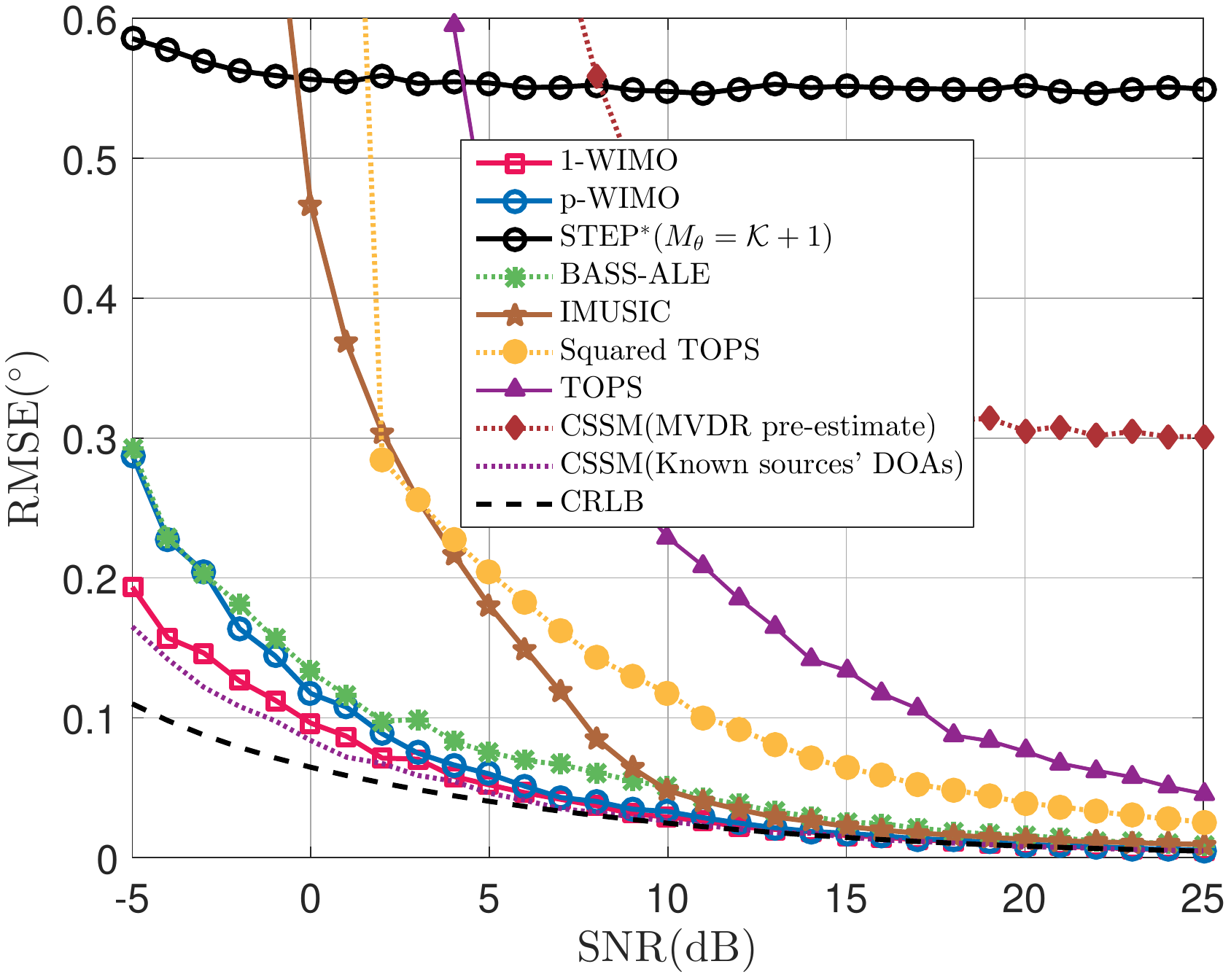}
	\caption{DOA estimation RMSE versus SNR. Two sources with bandwidth ratio $\eta=100\%$ are considered and 500 trials are averaged at each SNR.}
	\label{Fig.MSE_vs_SNR}
\end{figure}
In the next example, the probability of resolution of 1-WIMO and p-WIMO are compared with three wideband DOA estimators using sparse representation methods; $\ell_1$-SVD, W-CMSR and, W-LASSO. Simulation conditions are $f_l$=1.5KHz, $f_h$=4.5KHz with sources direction of arrival $15^\circ$ and $25^\circ$. It is observed in Fig.~\ref{Fig.SeperationProb_vs_SNR_Sparse} that 1-WIMO and p-WIMO outperform the other sophisticated methods in term of probability of resolution.
\begin{figure}
	\centering
	\includegraphics[width=0.64\linewidth]{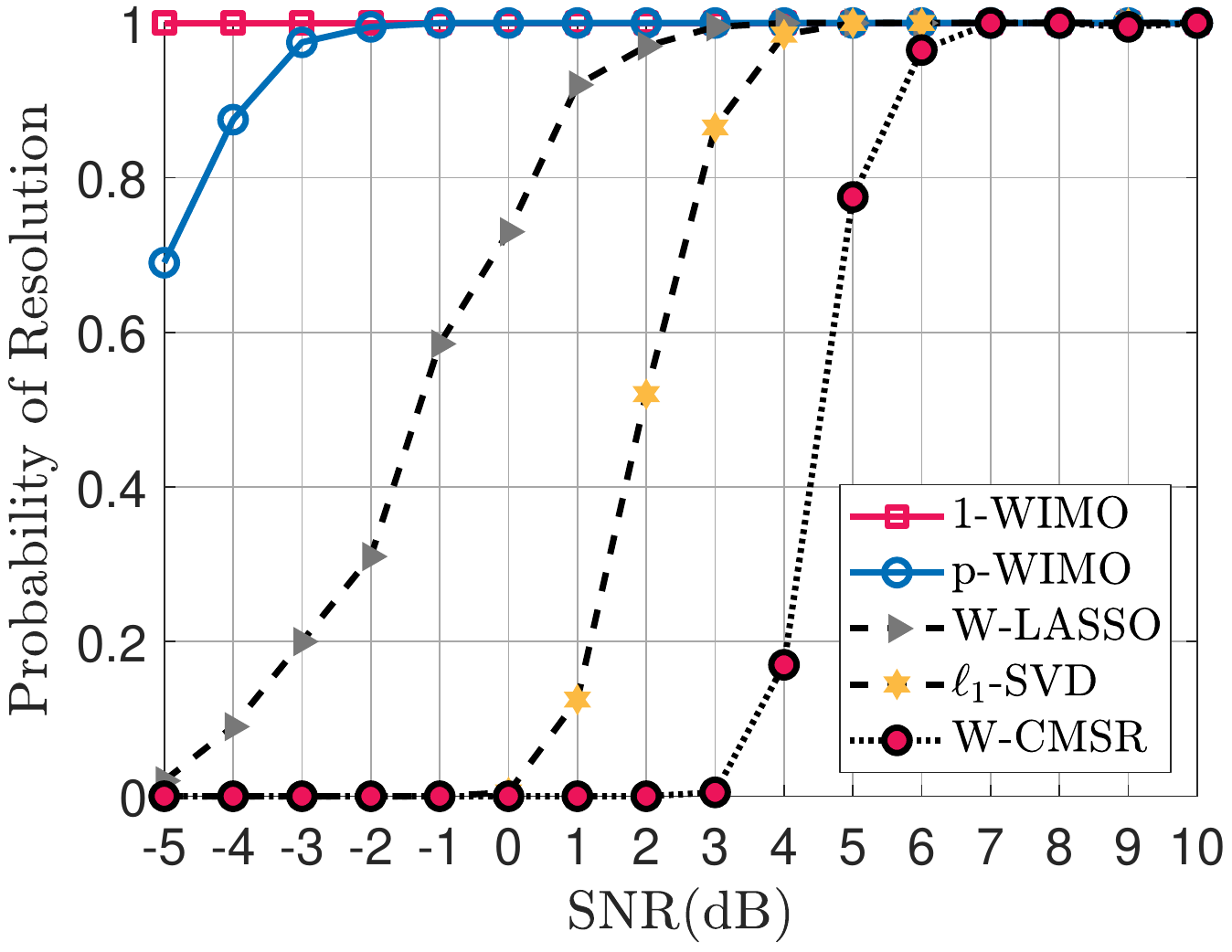}
	\caption{Probability of resolution against SNR for 1-WIMO and the three sparse representation methods $\ell_1$-SVD, W-CMSR and W-LASSO. 200 runs are averaged for each point.}
	\label{Fig.SeperationProb_vs_SNR_Sparse}
\end{figure}

WIMO performance in non-uniform power spectral density case is examined in the eighth numerical example. 
Two non-uniform PSD types Guassian and sinc$^2$ are assumed. The 3dB bandwidth for Gaussian and sinc$^2$ PSD are 1.2KHz and 0.7KHz respectively and the receiver bandwidth is set to 3KHz (1.5KHz$\sim$4.5KHz)\footnote{
	Receiver bandwidth larger than the input signal bandwidth is not usual, but to account for the worse case in term of increasing the signal subspace mismatch loss, receiver bandwidth is assumed larger than the source's 3dB power bandwidth.}. 
To explore the effect of mismatch loss in the case of uniform PSD assumption, WIMO with uniform spectrum formulation \eqref{equ.smat} is also applied with three different $B$ values. 
Simulation circumstances are the same as Fig.~\ref{Fig.SeperationProb_vs_SNR} (fifth example).
The input source's PSD and WIMO probability of separation versus SNR is illustrated in Fig.~\ref{Fig.SeperationProb_vs_SNR_nonuniform}. 
As expected, WIMO with knowing the non-uniform PSD of the signal leads to superior results, on the other hand, applying WIMO with uniform PSD assumption results in a loss due to GSV mismatch. 
Setting separation probability 0.9 as a benchmark, in Gaussian PSD, uniform assumption leads to 0.5dB, 0.8dB and 4.6dB loss for the three values of $B$ parameters. These values are 0.1dB, 0.4dB and, 8.7dB for sinc$^2$ PSD type. 

Consequently, in non-uniform spectrum situation, PSD information can be well exploited in WIMO to achieve the best performance. In addition, WIMO with uniform PSD assumption can be applied with acceptable performance loss, provided that bandwidth of this equivalent uniform PSD ($B$ in \eqref{equ.s_kl_p}) is set about the true PSD's 3dB power bandwidth. 
To compare the estimation error in the case of nonuniform PSD, we consider two BPSK sources with center frequency 3kHz and 3dB bandwidth 1.3kHz ($\eta$=44.5\%) located at $\left[-5^\circ,+5^\circ\right]$. The RMSE versus SNR is shown in Fig.~\ref{Fig.MSE_vs_SNR_BPSK}. An inconsistency is seen in some methods which stem from nonuniform sources' PSD and evenly utilization of the frequency bins information.
\begin{figure}
	\centering
	\includegraphics[width=0.64\linewidth]{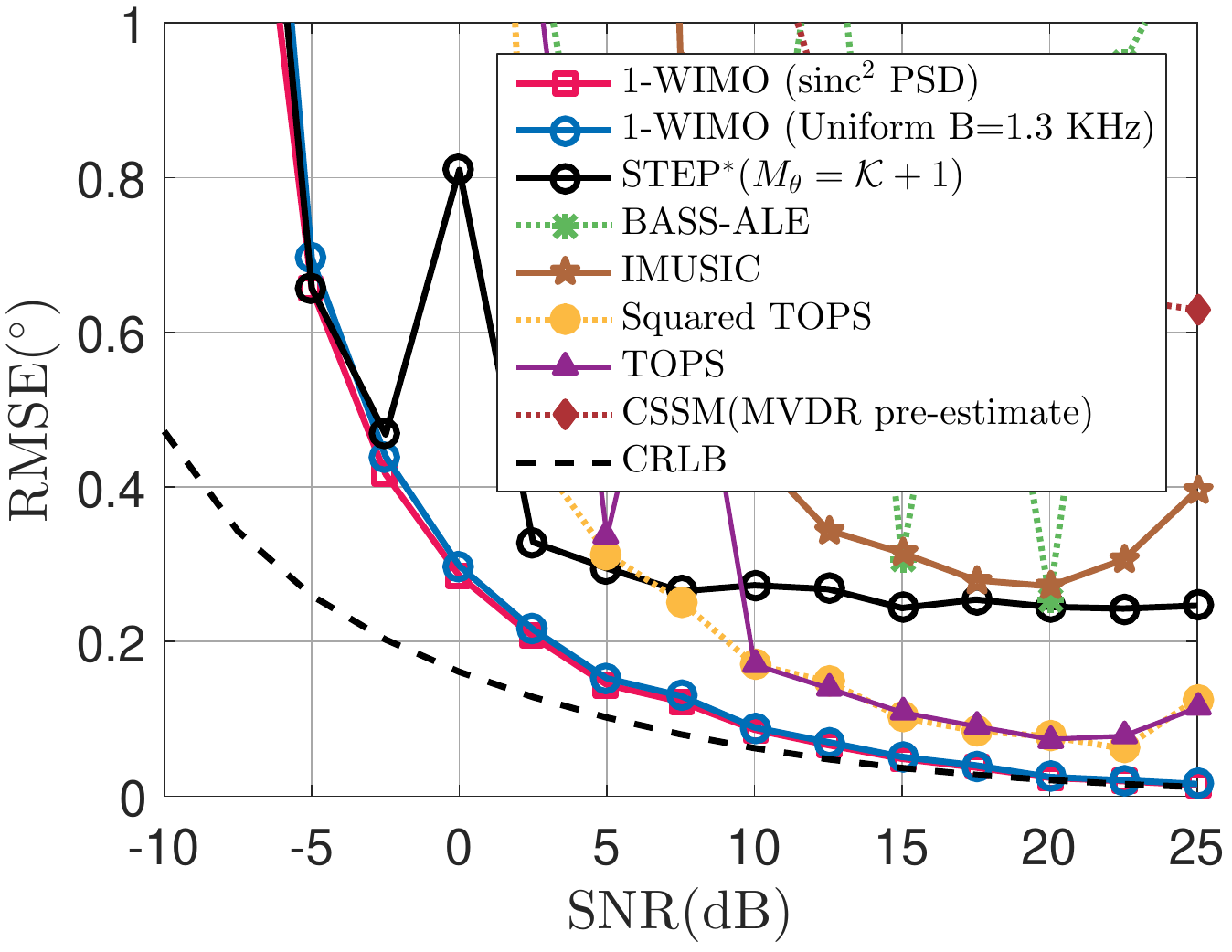}
	\caption{DOA estimation RMSE versus SNR for BPSK with $\eta$=44.5\%.}
	\label{Fig.MSE_vs_SNR_BPSK}
\end{figure}
\color{black}
\begin{figure*}
	\centering
	\subfloat[]{
		\includegraphics[width=0.25\linewidth]{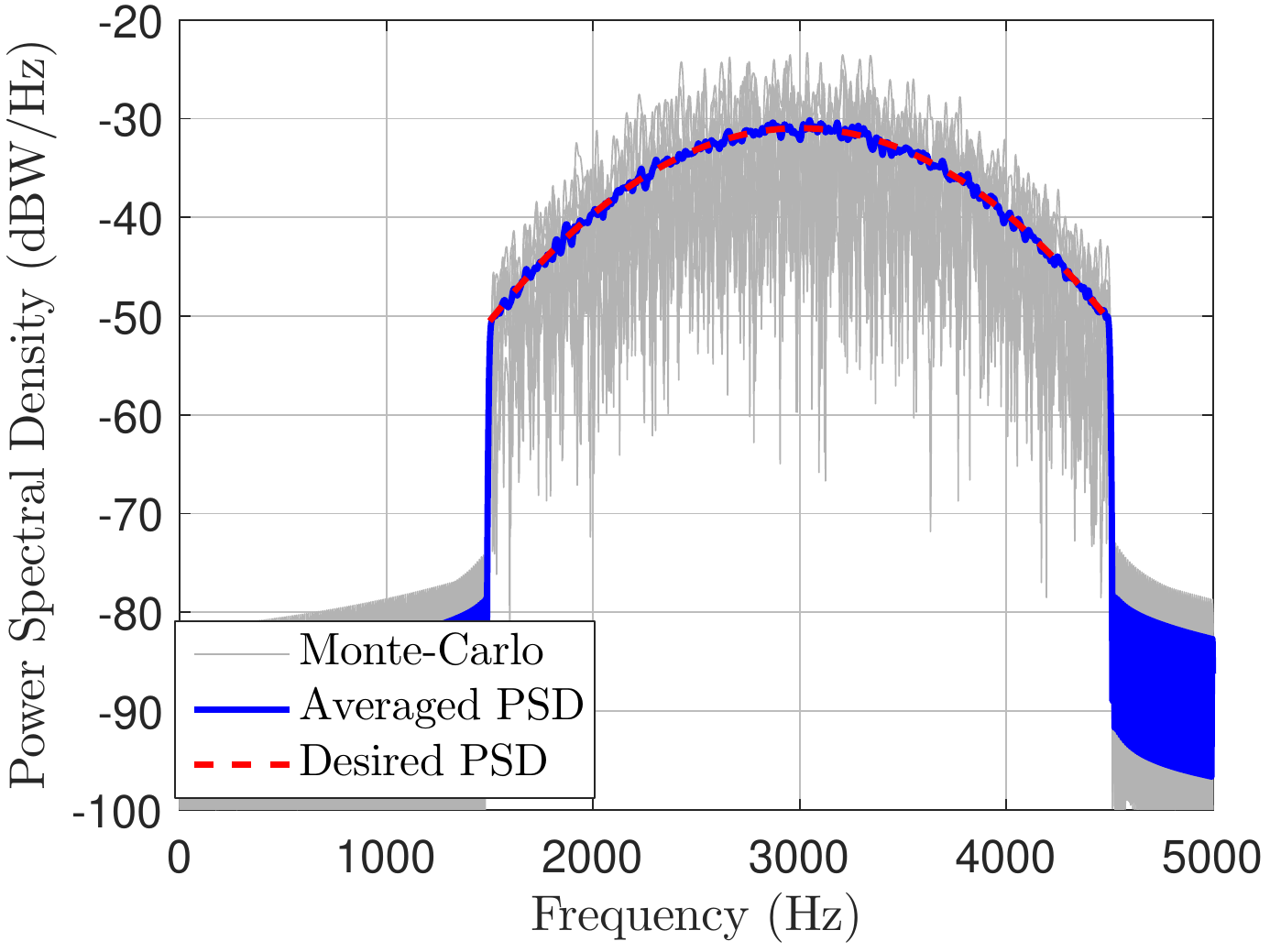}\label{subfig.gauss-psd}}
	\subfloat[]{
		\includegraphics[width=0.25\linewidth]{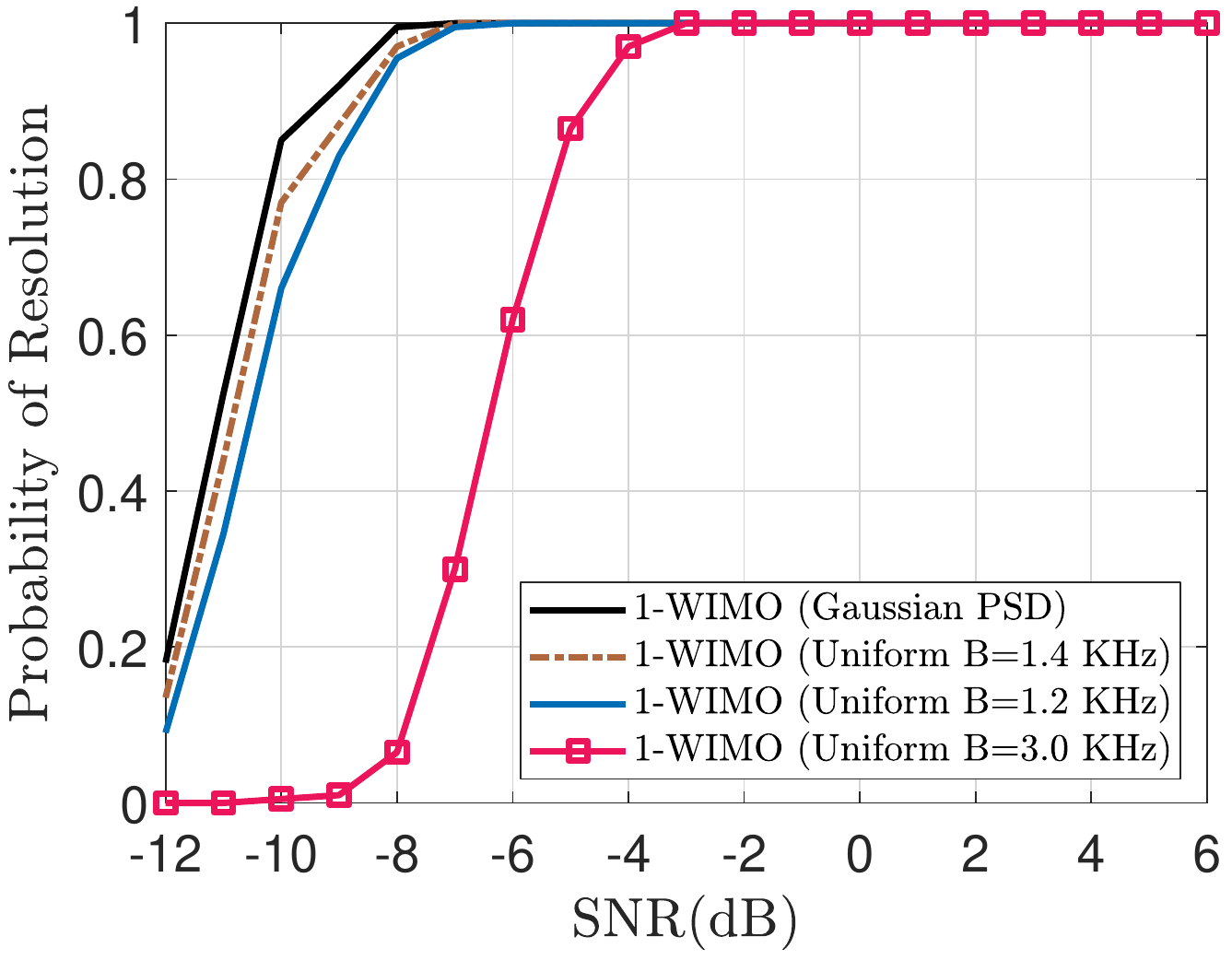}\label{subfig.pr-gauss}}
	\subfloat[]{
		\includegraphics[width=0.25\linewidth]{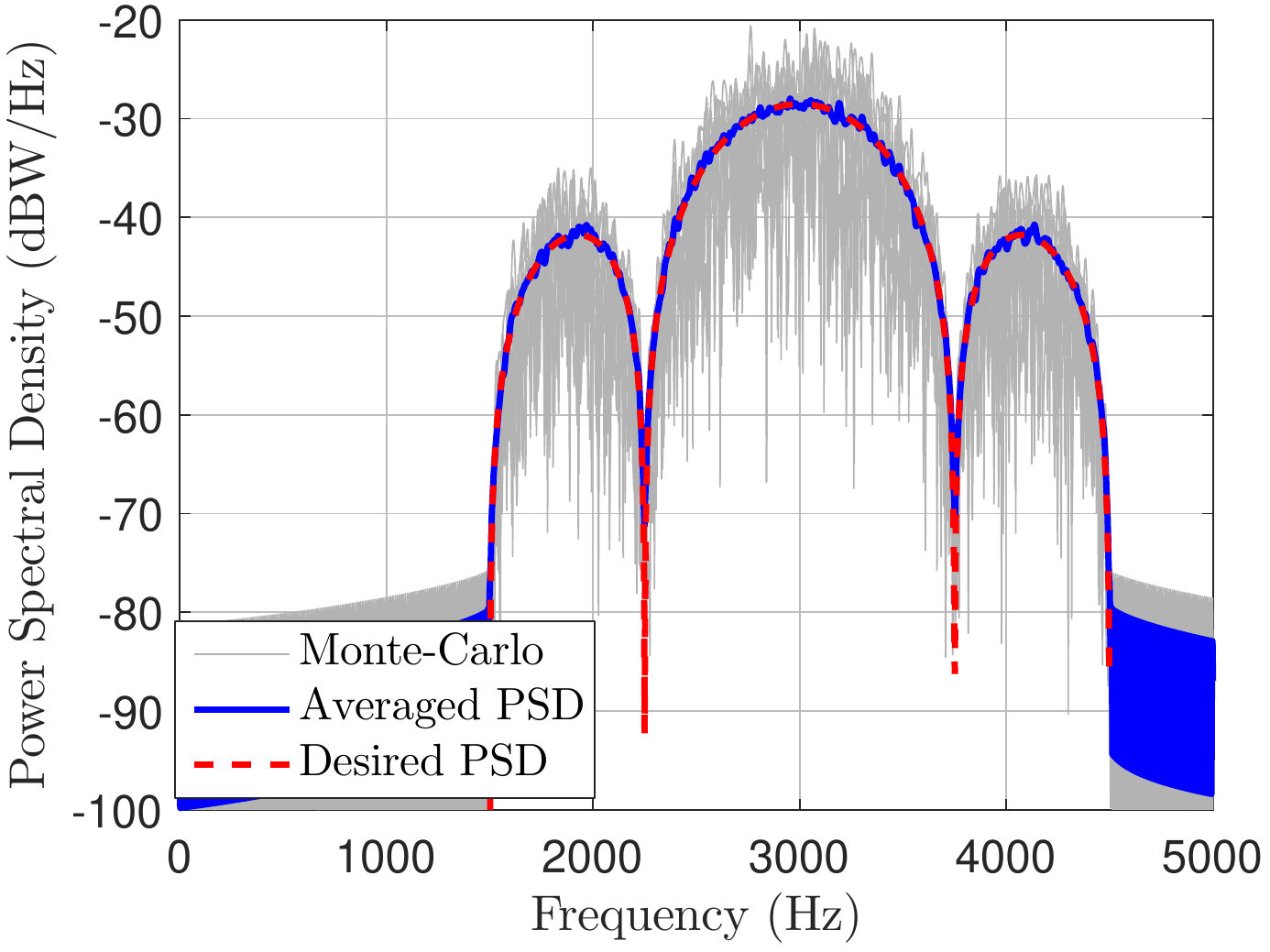}\label{subfig.sinc2-psd}}
	\subfloat[]{
		\includegraphics[width=0.25\linewidth]{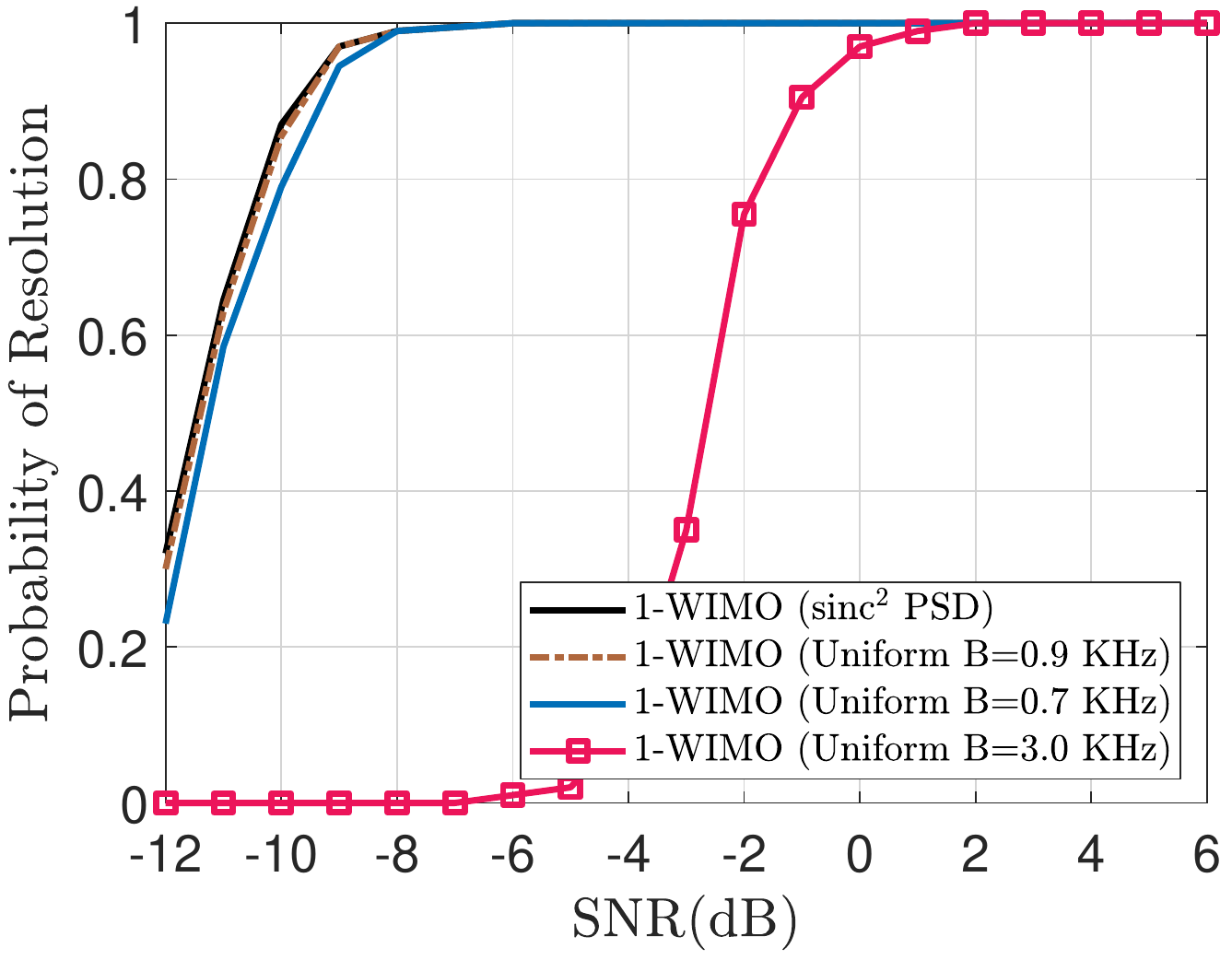}\label{subfig.pr-sinc2}}
	\hfil
	\caption{Probability of resolution versus SNR for two non-uniform PSD. 
	\protect\subref{subfig.gauss-psd} and \protect\subref{subfig.pr-gauss} are the PSD of the simulated signal and the resulting probability of resolution for Gaussian PSD respectively.
	\protect\subref{subfig.sinc2-psd} and \protect\subref{subfig.pr-sinc2} are for the sinc$^2$ PSD.}
	\label{Fig.SeperationProb_vs_SNR_nonuniform}
\end{figure*}

In the ninth example, the performance is examined versus the number of snapshots. In the previous numerical simulations the number of snapshot $M$ , was set to 8192, while in this simulation $M$ is swept from 100 to 10,000 samples. Two Gaussian distributed sources with SNR 5dB and $f\in [1.5kHz\sim 4.5kHz]$ are assumed and 200 runs are averaged for each snapshot. The results are depicted in Fig.~\ref{Fig.SeparationProb_vs_M}. It is seen that the superior resolution capability in low number of snapshots belongs to 1-WIMO. STEP$^*$ and IMUSIC are in the second rank. The number of sources $\mathcal{K}$ is assumed known in the STEP$^*$ and $M_\theta$ is set to $\mathcal{K}+1$.
\begin{figure}
	\centering
	\includegraphics[width=0.64\linewidth]{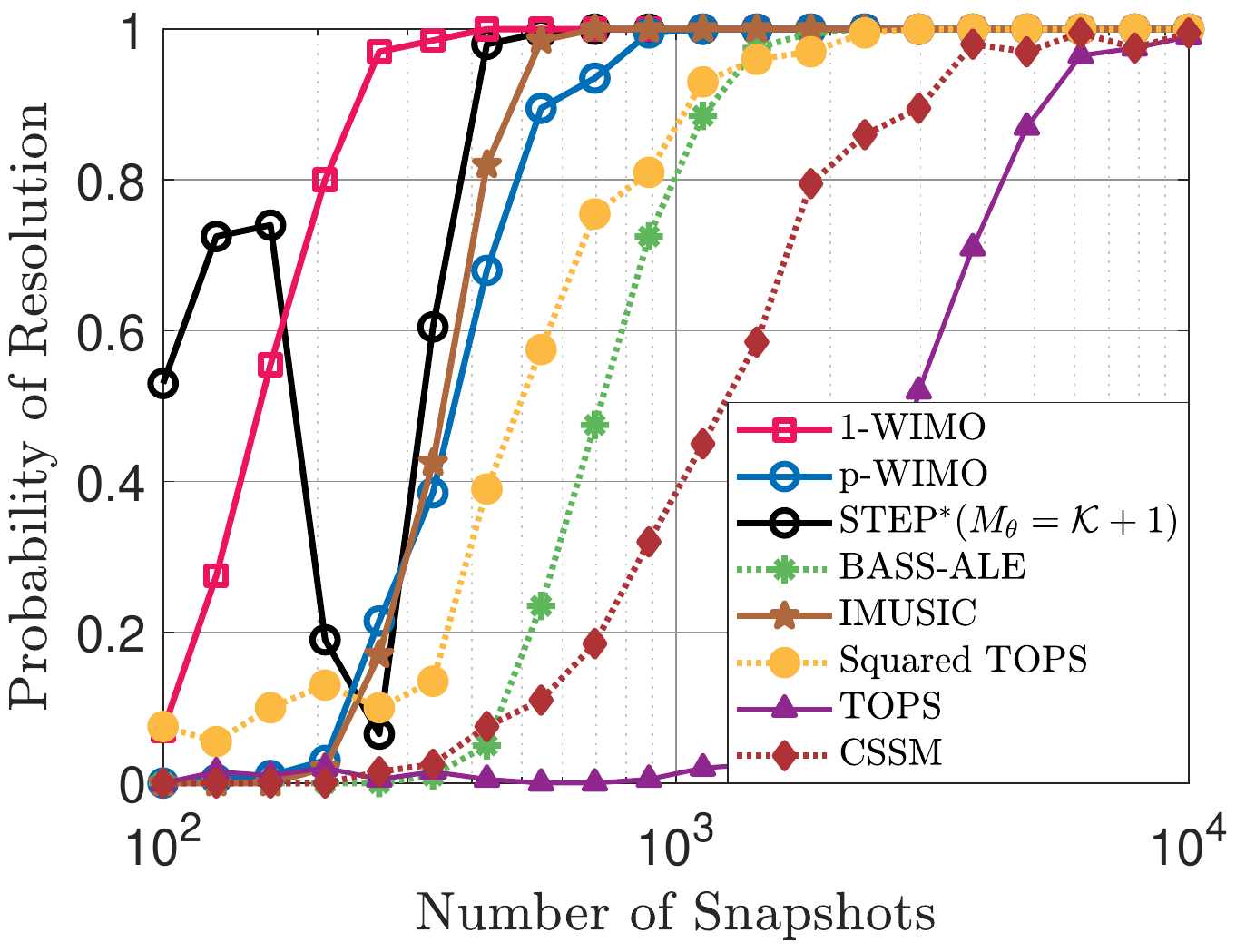}
	\caption{The probability of separation versus the number of snapshots.}
	\label{Fig.SeparationProb_vs_M}
\end{figure}
\color{black}



In the tenth example, we compare the DOA estimation RMSE of 1-WIMO with the two spatial-only approaches introduced in \cite{Agrawal2000}. The results are depicted in Fig.~\ref{Fig.RMSE_vs_SNR_SPO}. The simulation parameters are given in the figure's caption. $J_1(\theta)$, $\tilde{J}_1(\theta)$ and $J_2(\theta)$ are based on equations (18), (20) and (25) of \cite{Agrawal2000} respectively. Avoiding the infeasible exhaustive search through all source's number for $J_1(\theta)$ and $\tilde{J}_1(\theta)$, the true number of sources are supplied in these methods. For $J_2(\theta)$ eigenvectors corresponding to noise subspace is manually set to the eigenvectors corresponding to the three smallest eigenvalues of SCM. It is seen that with given sources' number, $\tilde{J}_1(\theta)$ has better performance in the low SNR regime but 1-WIMO surpasses it in medium and high SNR values. It is noteworthy that for $\theta\in\left[-90^\circ\sim +90^\circ\right]$ with grid size $1^\circ$, 1-WIMO is 8 times slower than $J_2(\theta)$ but it is about 120 times and 900 times faster than $J_1(\theta)$ and $\tilde{J}_1(\theta)$ respectively
\footnote{
	The computational complexity of $J_1(\theta)$ and $\tilde{J}_1(\theta)$ are $\mathcal{O}(N_\theta^\mathcal{K})$, i.e., increase exponentially with the number of sources $\mathcal{K}$.}
\begin{figure}
	\centering
	\includegraphics[width=0.95\linewidth]{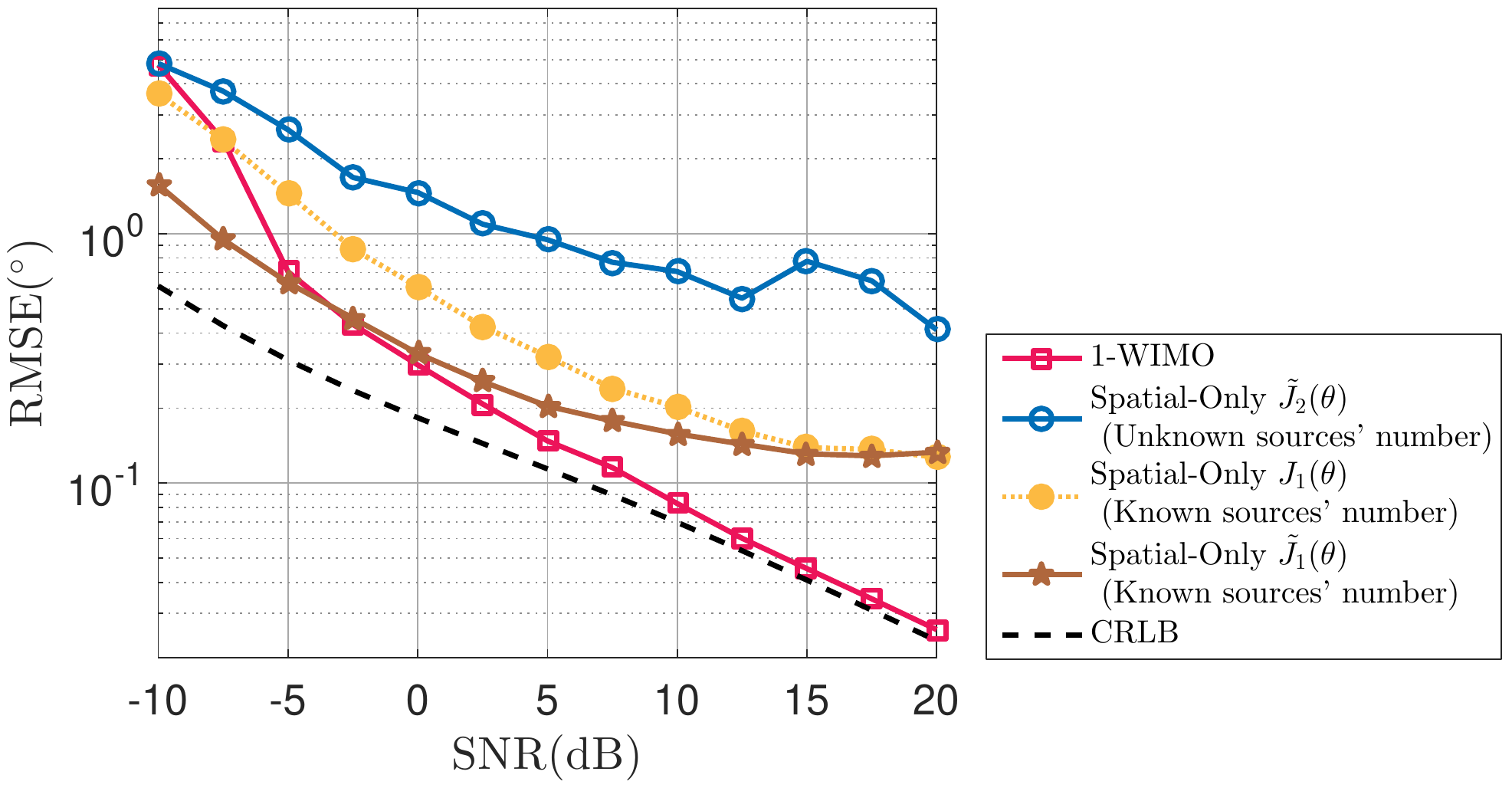}
	\caption{DOA estimation RMSE versus SNR for 1-WIMO and spatial-only methods \cite{Agrawal2000}. 
	ULA with $N_S$=8, M=1024, $f_l$=1.5KHz and $f_h$=4.5KHz with uniform PSD, $f_s$=10KHz and sources' DOAs are $\left[-5^\circ,+5^\circ\right]+\nu$ and 500 trials are averaged for each SNR.}
	\label{Fig.RMSE_vs_SNR_SPO}
\end{figure}

\color{MyBlue}
In the next example, we compare the resolution capability of the methods versus sources' DOA separation. Two Gaussian sources with $\eta$=60\%, SNR=0dB and 1024 snapshots are assumed. Sources are located at $\left[20-\frac{\Delta\theta}{2} , 20+\frac{\Delta\theta}{2} \right]$ and successful separation is attributed to a spectrum with minimum prominence 3dB and maximum error $1^\circ$. The result is illustrated in Fig.~\ref{Fig.SeperationProb_vs_Theta}, which shows STEP (with known sources' number) and 1-WIMO have superior resolution capability in low SNR and adjacent sources.
\begin{figure}
	\centering
	\includegraphics[width=0.64\linewidth]{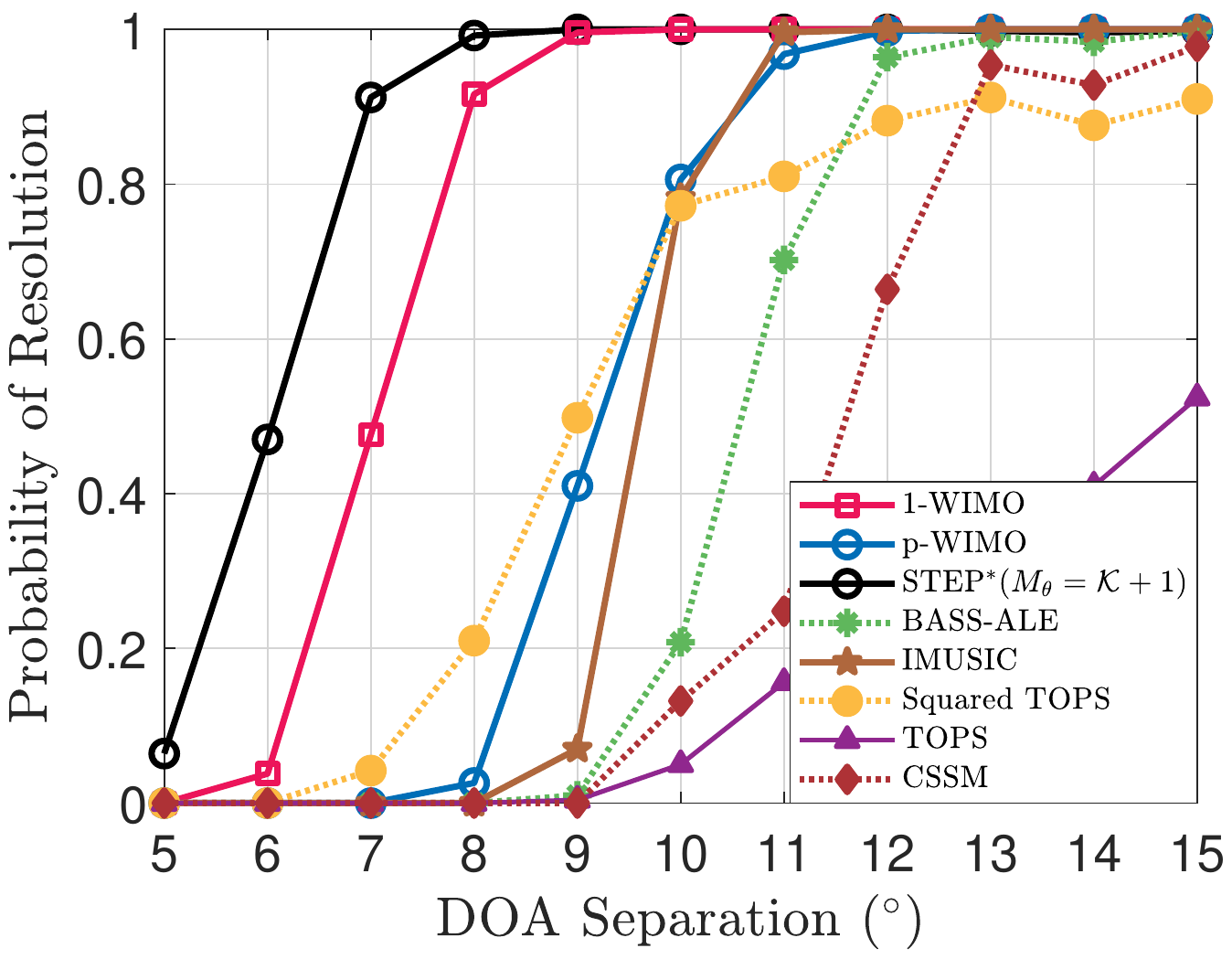}
	\caption{Probability of resolution versus sources' DOA separation. 500 trials are run for each SNR.}
	\label{Fig.SeperationProb_vs_Theta}
\end{figure}
\color{black}

\color{MyBlue}
Resolution sensitivity to the sources' coherence is studied in the next numerical example. Two wideband sources with $\eta$=60\%, SNR=5dB are located at 15$^\circ$ and 25$^\circ$. There are 1024 samples with $f_s$=10kHz. 500 trials are run at each SNR and success condition is the same as the previous example. The resolution probability versus coherence index $\rho$ is shown in Fig.~\ref{Fig.SepProb_vs_Coho}. Obviously $\rho$=0 corresponds to sources' independence and $\rho$=1 denotes fully correlated sources. 1-WIMO shows the best performance against sources' correlation and Squ-TOPS, p-WIMO and IMUSIC are in the second rank.
\begin{figure}
	\centering
	\includegraphics[width=0.64\linewidth]{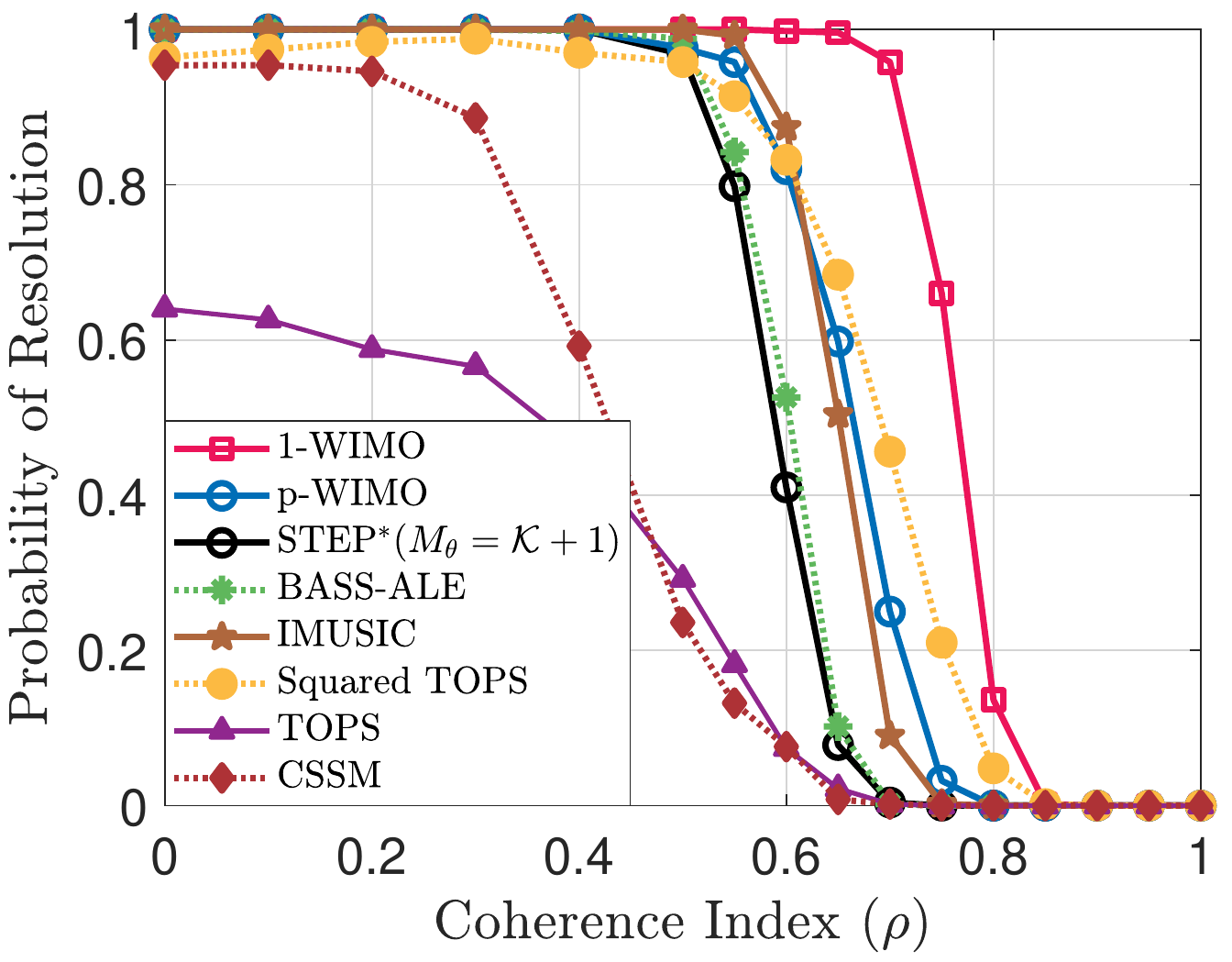}
	\caption{Probability of resolution versus sources' coherence index.}
	\label{Fig.SepProb_vs_Coho}
\end{figure}
\color{black}

In the last numerical example, we compare methods runtime versus bandwidth ratio. $\eta$ is swept from 0\% to 150\%, $f_h$=4KHz
and M=8192.
\color{black}
 Runtime is reported as the averaged calculation time of spatial spectrum for $-90^\circ$ to $+90^\circ$ interval with grid size $1^\circ$. $\Delta f$ for the methods involving subband processing is set to $50$Hz; namely CSSM, IMUSIC, TOPS, squared TOPS, STEP$^*$ and $\ell_1$-SVD. The results are illustrated in Fig.~\ref{Fig.RunTime_vs_B_with_ACOV}.
\begin{figure}[!t]
	\centering
	\includegraphics[width=3in]{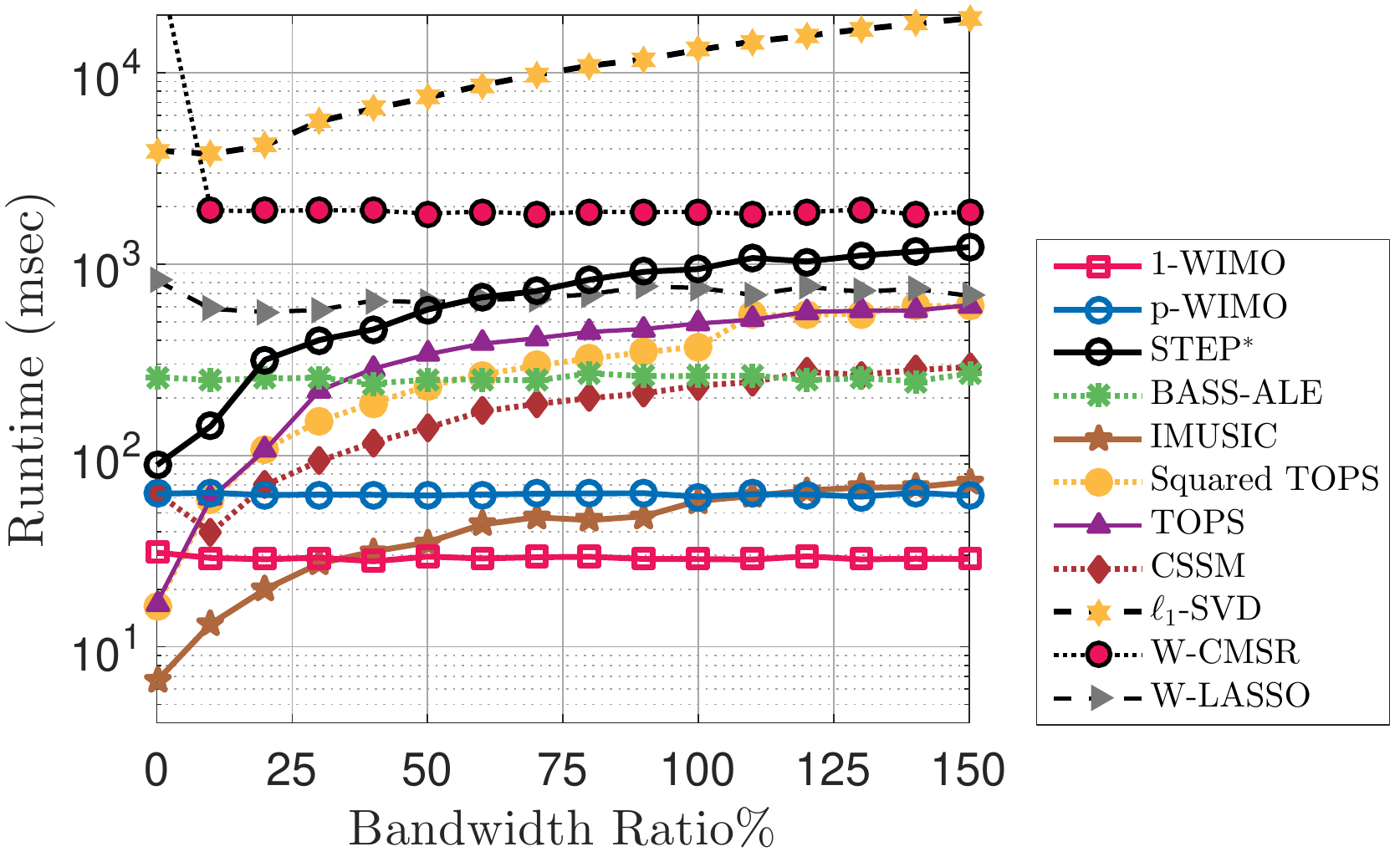}
	\caption{Comparison of the averaged runtime versus bandwidth ratio.}
	\label{Fig.RunTime_vs_B_with_ACOV}
\end{figure}
$\ell_1$-SVD and W-LASSO utilizes CVX toolbox and W-CMSR uses SeDuMi \cite{sturm1999} as the solver. This simulation is run on a PC with the following specifications: Windows 10 (64bit version), Intel Core i7-920 3.4 GHz and 8GB RAM. 

Numerical results, show that the main part of calculation of WIMO belongs to estimation and decomposition of $\hat{\mathbf{S}}_{\tilde{\mathbf{y}}\tilde{\mathbf{y}}}$, which increases with $\mathcal{O}(m^3)$. Therefore, regarding the previously discussed constraint on $\varepsilon$, $m$ can be selected as small as possible if a faster calculation is required. 
It is also shown that in identical conditions, BASS-ALE has more computational cost due to the spatial averaging process over multiple subbands.

\section{Conclusion}\label{sec.conclusion}
A new broadband DOA estimation approach, named wideband modal orthogonality (WIMO) was introduced in this paper. WIMO proposes a subspace-based solution, wherein the noise subspace is obtained from eigen-analysis of STCM and the signal subspace is provided from a mathematical approximation of STCM.
Our proposed DOA estimators, namely 1-WIMO and p-WIMO, overcome 
\color{MyBlue} some of \color{black}
the implementation challenges of the previous methods, such as spectral decomposition, focusing procedure and, multiple EVD for all subbands and test angles. They also require no a priori information on the number of sources and the pre-estimates of their DOAs.
Extensive numerical investigations demonstrated the superior performance of WIMO with more than one order of magnitude faster runtime compared to some state-of-the-art schemes.
We also showed that the presented STCM approximation can accurately estimate broadband signal subspace dimension.
In the case of non-uniform power spectral density, it is shown that knowledge about source's spectral distribution can be well exploited to achieve the optimum WIMO performance. On the other hand, uniform assumption in the signal bandwidth can provide satisfactory results.
\appendix
\section{Proof of theorem~\ref{the.1}} \label{app.0}
\begin{proof}
	 Using the \textit{Rayleigh quotient theorem} \cite{Horn2013}, the $i$-th eigenvector can be formulated as the solution of the following optimization,
	\begin{align}
	&\breve{\mathbf{u}}_i = \underset{\mathbf{u}}{\text{argmax}} \; \sum_{k,l}u_k^*\breve{s}_{k,l}u_l \label{equ.Rayleigh2}\\
	& \text{s.t.} \; \|\mathbf{u}\|_2 = 1 \; , \; \mathbf{u}^H \breve{\mathbf{u}}_j=0 \; , \; 
	\lbrace i,j \in \lbrace 1,\cdots,L \rbrace \; , \; j<i \rbrace \nonumber
	\end{align}
	From \eqref{equ.sn_sw} $\breve{\mathbf{S}}^N = \mathbf{g}\mathbf{g}^H$, then $\text{rank}(\breve{\mathbf{S}}^N) = 1$ and $\breve{\mathbf{u}}_1^N = \mathbf{g}/\sqrt{L}$. substituting $\breve{s}^W_{k,l}$ and $\breve{s}_{k,l}$ by the following integrals,
	\begin{align}
	\breve{s}^W_{k,l} &= \frac{1}{B} \int_{-\frac{B}{2}}^{+\frac{B}{2}} e^{j2\pi (h_k - h_l)u}\cdot du
	\\
	\breve{s}_{k,l} &= \frac{1}{B} \int_{-\frac{B}{2}}^{+\frac{B}{2}} e^{j2\pi (h_k - h_l)(u + f_c)}\cdot du
	\end{align}
	and then rewriting the objective in \eqref{equ.Rayleigh2}, we obtain,
	\begin{equation} \label{equ.sigma_i}
	\begin{aligned}
	\sum_{k,l}u_k^*\breve{s}_{k,l}u_l = \frac{1}{B}\int_{-\frac{B}{2}}^{+\frac{B}{2}} \left| \sum_{k} u_k e^{-j2\pi h_k(u+f_c)} \right|^2 \cdot du\\
	= \frac{1}{B}\int_{-\frac{B}{2}}^{+\frac{B}{2}} \left| \sum_{k} u^\prime_k e^{-j2\pi h_k u} \right|^2 \cdot du
	= \sum_{k,l}u_k^{\prime *}\breve{s}^W_{k,l}u_l^\prime
	\end{aligned}
	\end{equation}
	where in the last equality, there is a change of variable from $\mathbf{u}$  to $\mathbf{g} \circ \mathbf{u}^\prime$. If $\mathbf{u}^\prime = \breve{\mathbf{u}}_i^W$ then $\mathbf{u}$ maximizes \eqref{equ.sigma_i} and for $j<i$, we have,
	\begin{equation}
	\mathbf{u}^H \breve{\mathbf{u}}_j = 
	(\mathbf{g}\circ\breve{\mathbf{u}}_i^W)^H(\mathbf{g}\circ\breve{\mathbf{u}}_j^W) = 0
	\end{equation}
	then $\breve{\mathbf{u}}_i=\mathbf{g}\circ\breve{\mathbf{u}}_i^W$. Furthermore, the equality of $\sigma_i(\breve{\mathbf{S}})$ and $\sigma_i(\breve{\mathbf{S}}^W)$ and their positive semi-definiteness is obtained from \eqref{equ.sigma_i}.
		Given that $S(f)\ge 0$ for all $f$ (see (9-164) in \cite{Papoulis2002a}), regarding \eqref{equ.int_approx}, positive semi-definiteness of $\breve{\mathbf{S}}$ holds also in general non-uniform case.
	Denoting $\breve{U}^W_{\mathcal{SS}^m}$ as the SS-Transform of wideband part of the GSV, expression for $\sigma(\breve{\mathbf{S}})$ in \eqref{equ.sigma_i}, is in-fact integration of the power of the $\breve{U}^W_{\mathcal{SS}^m}$ over the whole bandwidth,
	\begin{equation}\label{equ.Jx2}
	\sigma_1(\breve{\mathbf{S}}) = 
	\int_{-\frac{B}{2}}^{+\frac{B}{2}} 
	\left| \breve{U}^W_{\mathcal{SS}^m}(u,\theta) \right|^2 \cdot du
	\end{equation}
\end{proof}
\bibliographystyle{IEEEtr}

\end{document}